\documentclass[10pt,journal,compsoc]{IEEEtran}
\usepackage{amsmath, amsfonts, amsthm, amssymb}
\usepackage[linesnumbered, ruled, vlined]{algorithm2e}
\usepackage{array}
\usepackage[caption=false,font=small,labelfont=rm,textfont=rm]{subfig}
\usepackage{textcomp}
\usepackage{stfloats}
\usepackage{url}
\usepackage{verbatim}
\usepackage{graphicx}
\usepackage{cite}
\usepackage{color}
\usepackage[colorlinks=true, linkcolor=black, anchorcolor=blue, citecolor=black]{hyperref}
\usepackage{longtable}
\usepackage{nomencl}
\usepackage[switch, mathlines]{lineno}
\def\BibTeX{{\rm B\kern-.05em{\sc i\kern-.025em b}\kern-.08em
    T\kern-.1667em\lower.7ex\hbox{E}\kern-.125emX}}

\graphicspath{{fig}}
\makenomenclature
\allowdisplaybreaks[4] 

\theoremstyle{remark} 

\newtheorem{proposition}{Proposition}

\newtheorem{lemma}{Lemma}

\renewcommand{\eqref}[1]{(\ref{#1})}
\newcommand{\figref}[1]{Fig.~\ref{#1}}
\newcommand{\tabref}[1]{Tab.~\ref{#1}}
\newcommand{\algref}[1]{Alg.~\ref{#1}}
\newcommand{\alglinref}[1]{Line~\ref{#1}}
\newcommand{\secref}[1]{Section~\ref{#1}}
\newcommand{\lemref}[1]{Lemma~\ref{#1}}

\newcommand{\propref}[1]{Propositon~\ref{#1}}

\begin{document}
\title{User Association and Channel Allocation in 5G Mobile Asymmetric Multi-band Heterogeneous Networks}
\author{Miao Dai, \textit{Member, IEEE}, Gang Sun, \textit{Senior Member, IEEE}, Hongfang Yu, \textit{Member, IEEE}, \\ Sheng Wang, \textit{Member, IEEE}, and Dusit Niyato, \textit{Fellow, IEEE}
\thanks{This work was supported in part by the National Key Research and Development Program of China under Grant 2019YFB1802800. 
	}
\thanks{Miao Dai, Gang Sun, Hongfang Yu, and Sheng Wang are with the Key Laboratory of Optical Fiber Sensing and Communications (Ministry of Education), University of Electronic Science and Technology of China, Chengdu 611731, China (e-mail: daimiao@std.uestc.edu.cn; gangsun@uestc.edu.cn; yuhf@uestc.edu.cn; wsh\_keylab@uestc.edu.cn).}
\thanks{Dusit Niyato is with College of Computing and Data Science, Nanyang Technological University, Singapore (e-mail: dniyato@ntu.edu.sg).}

\vspace*{5pt}
\parbox{0.9\textwidth}{
\begin{flushleft}
	\begin{abstract}
		
		With the proliferation of mobile terminals and the continuous upgrading of services, 4G LTE networks are showing signs of weakness. To enhance the capacity of wireless networks, millimeter waves are introduced to drive the evolution of networks towards multi-band 5G heterogeneous networks. The distinct propagation characteristics of mmWaves and microwaves, as well as the vastly different hardware configurations of heterogeneous base stations, make traditional access strategies no longer effective. Therefore, to narrowing the gap between theory and practice, we investigate the access strategy in multi-band 5G heterogeneous networks, taking into account the characteristics of mobile users, asynchronous switching between uplink and downlink of pico base stations, asymmetric service requirements, and user communication continuity. We formulate the problem as integer nonlinear programming and prove its intractability. Thereby, we decouple it into three subproblems: user association, switch point selection, and subchannel allocation, and design an algorithm based on optimal matching and spectral clustering to solve it efficiently. The simulation results show that the proposed algorithm outperforms the comparison methods in terms of overall data rate, effective data rate, and number of satisfied users.
		
	\end{abstract}
	
	
\end{flushleft}}
\vspace*{-40pt}
}


\maketitle

\section{Introduction} \label{sec:introduction}

\IEEEPARstart{W}{ith} the continuous popularization and deepening of the network, the number of mobile terminals entering the network is growing exponentially, and people are putting forward more and more diversified and customized demands for network services. In order to adequately respond to the challenges posed to the network by the growing number and differentiated requirements of mobile subscribers, the radio access network (RAN), which serves as the entry for the service, is evolving from the homogeneous 4G long-term evolution (LTE) network to the 5G heterogeneous network (HetNet) \cite{odida2024evolution} where multiple macro and pico base stations (BS) coexist. These BSs are different in size, transmit power, carrier frequency, and coverage, and the former is responsible for providing wide-area signal coverage, while the latter are densely deployed in hotspots or in areas where the former's coverage is weak, to complement the wireless network and enhance its service capability.

The microwave band used by LTE BSs is a scarce and exhausted resource that can hardly accommodate the surge in wireless traffic. Therefore, 5G HetNets are simultaneously developing the millimeter-wave (mmWave) band \cite{ashraf20245g, xue2024survey}, whose bandwidth is far beyond the microwave band. Moreover, mmWave exhibits directional propagation characteristics, which can enhance the beamforming capability for wireless signals when combined with directional antennas. However, mmWave has higher propagation loss and poorer ability to bypass obstacles and is therefore sensitive to whether or not the communicating entities are at line-of-sight (LOS) to each other.

The iterative updating of network services and the resulting evolution of network infrastructure are announcing that traditional access strategies are no longer able to achieve the desired results in new scenarios. Specifically, the reference signal receiving power (RSRP) and received signal strength indicator (RSSI) are not a good reflection of the reachable communication rate in 5G HetNets \cite{semiari2016downlink}, and the disparity in transmit power between macro and pico BSs can easily lead to unbalanced association results in user association (UA) strategies based on these metrics \cite{khawam2020coordinated, liu2016user}. What's more, the directional propagation properties of mmWave change the interference structure from being ubiquitous to being dependent on UA results \cite{alizadeh2022reinforcement, liu2020user, alizadeh2020multi, alizadeh2019load}. On the other hand, traditional networks always assume that users would primarily receive data from the BS, many emerging services break this convention. For example, live streaming has a greater need for stable and sufficient uplink transmissions, while video calls and conferencing require both uplink and downlink data rates. The asymmetry of the uplink and downlink requirements of services makes the UA strategies originally designed for the downlink only no longer effective: they may associate the user to a far-away macro BS, where the user's transmit power is not sufficient to guarantee a minimum data rate in uplink direction considering its hardware limitation \cite{luo2014downlink}. In addition, conventional LTE BSs require synchronization of uplink and downlink switching points, otherwise, the disparity in transmit power between the BS and the mobile user equipment (MUE) may lead to irreconcilable cross-link interference. However, this is not the case for pico BSs in HetNets, as their transmit power is very close to that of the MUE, which allows pico BSs to independently adjust their uplink and downlink switching points to better adapt to the services they carry \cite{pedersen2016flexible, xu2023energy}. These new features provide strong grips on scheduling the network at a finer granularity but also bring challenges that more complex factors inevitably have to be taken into account when designing scheduling policies.

Since the proposal of 5G HetNets, how to adjust the access strategy to mitigate the above problems has aroused extensive investigations in the academic community. In this paper, we will analyze and compare the main ideas and limitations of previous research, explore the wireless access problems that still exist in the new network paradigm, and design effective solutions to drive 5G HetNet toward the vision of the Internet of Everything.

\subsection{Related Work} \label{sec:related work}

Some researchers focus on primitive HetNets where BSs differ in size and transmit power but still use the same communication bands. The authors in \cite{kim2023distributed} improved the minimum downlink data rate of the user through UA and bandwidth allocation (BA). They used global channel state information (CSI) to design a Lagrangian primal-dual approach that optimizes UA and BA alternatively. Additionally, deep reinforcement learning (DRL) was utilized to determine the dual variables to reduce the information exchange between BSs. A three-tier HetNet is considered in \cite{ye2013user}. UA and BS time slot allocation are performed to optimize a logarithmic utility function defined over the user's downlink data rate, which is a combinatorial optimization problem. The authors proposed a centralized method that turns it into convex optimization by relaxing the UA variables, and then solving it and rounding the result. A distributed UA method based on biased user data rates is also presented here, where the bias factors are obtained from Lagrangian dual decomposition. Although there is still some studies grounded in this scenario \cite{tam2016joint, fang2020energy, li2019user}, the scarce radio resources in the microwave bands make it difficult for this paradigm to become mainstream in 5G networks.

mmWave is believed to be the inevitable evolutionary direction for future networks. The authors in \cite{zhang2020optimal} explored the UA problem in mmWave networks to minimize the proportion of unassociated users under the constraint of BS connection capacity. The authors treated the interference of mmWave communications as zero and assume that users are satisfied as long as they are within the coverage area of the BS and in LOS. Based on this setting, the authors constructed a bipartite between users and BSs and find the maximum match of the graph to provide connections to more users. To reduce the complexity, the authors also proposed a heuristic method: Starting with users surrounded by fewer available BSs, the algorithm associates them with the least loaded one among the available BSs. The authors in \cite{alizadeh2019load} focused on the UA problem in dense mmWave networks, with the objective of optimizing the system downlink data rate, taking into account the user's single connection and BS data stream capacity constraints. They assumed that perfect instantaneous CSI is available, and the proposed algorithm continuously swaps the BS of the UE with the worst connection quality with the BSs of other UEs until the association scheme remains unchanged after a certain number of consecutive swaps. Unlike these studies above, which only consider a single connection for users, the authors in \cite{alizadeh2019load} assumed that each user can establish multiple connections at the same time. In this case, the authors modeled the UA problem of maximizing the downlink data rate as a multi-label classification issue and transformed it into a series of single-label classification problems. Subsequently, a graph model is constructed based on user location and network topology information, whose feature vectors are extracted and then classified by a machine learning classifier. These studies have pushed the exploration of mmWave communications, but the overly optimistic idea of ignoring mmWave interference in \cite{zhang2020optimal}, and the difficulty of realistically capturing the full CSI needed by \cite{alizadeh2019load}, as well as the fact that all of them have only considered mmWave BSs, are issues that make them still somewhat different from a future 5G HetNet.

The majority of the investigations are based on multi-band wireless scenarios where distinct BSs coexist. The authors in \cite{qiao2023joint} optimized UA and spectrum allocation (SA) to maximize the total uplink data rate for a multi-band HetNet, where the bands between small BSs are also orthogonal and the users here are immobile. The authors form users at the same BS into a coalition and traverse all users to switch coalitions for them so that the total uplink data rate of the related coalitions is increased until no more switching can be done. A three-tier HetNet is considered in \cite{khawam2020coordinated}, in which users are also immobile. The authors formulate the UA problem of maximizing the downlink rate given a SA result as a potential game, which is guaranteed to converge to an equilibrium point. The authors in \cite{liu2019joint} optimized the UA, power allocation (PA), and time slice allocation with the objective of weighted proportional fairness of downlink throughput, taking into account the single connection of users, the maximum transmit power of BSs, and the time slice demand of users as constraints. It is transformed into an alternating optimization of UA and PA under the assumption that mmWave interference is still omnidirectional. Thereby, the optimal UA is solved using the Lagrangian primal-dual by fixing the PA decision, while the PA is solved using the Newton-Raphson method by fixing the UA decision.

In addition to these mathematical optimization methods, machine learning based methods are also used. In \cite{deng2023gnn}, the interference relationship between users and nearby BSs is abstracted into an interference graph, and then the graph neural network is used to output the corresponding UA and beam selection decisions using it as input. The authors in \cite{sana2020multi} used a multi-agent deep reinforcement learning (MADRL) based UA method to enhance the downlink data rate of the system. Each user acts as an independent agent and makes association decisions based on local observations, and if ultimately there is a conflict in the association decisions, all agents are rewarded with zero or the total data rate otherwise. A similar approach appears in \cite{guo2020joint}, the authors adopt MAPPO to make UA decisions and power requests to promote the downlink throughput, which also acts as the reward for all agents. Since all agents use the same reward, both \cite{sana2020multi} and \cite{guo2020joint} theoretically require cooperative network environments. The authors in \cite{alizadeh2022reinforcement} proposed a semi-distributed UA method based on the multi-armed bandit (MAB). Each user maintains an MAB and selects its own association decision based on the upper confidence bound principle, and then the central controller aggregates the association decisions of all the users to form a temporary association scheme. The downlink data rates yielded by the temporary and historically optimal association schemes are compared, and the scheme that yields a larger downlink rate is adopted as the actual association scheme. Unlike the previously presented works, the authors in \cite{chaieb2022deep} additionally considered a minimum data rate requirement for the service. The authors use MADQL for decision-making, where the BS replaces the user to be an agent, and the reward of each agent is defined as its own total downlink data rate. These works have made valuable research on access strategies for 5G multi-band HetNets, but there are still some limitations. For example, most of them have not required the minimum data rate of the services, they have neither considered the coexistence of differentiated services nor do they consider the asymmetric uplink and downlink demands of evolving services.

\subsection{Motivation and Contributions} \label{sec:motivation and contributions}

In this paper, we aim to build on previous research and further bridge the gap between theory and reality. We focus on the wireless access problem for mobile users in 5G HetNets where macro and pico BSs coexist. They operate in different frequency bands, and pico BSs are allowed to use asynchronous uplink and downlink switching points. In addition, mobile users have differentiated and asymmetric uplink and downlink service requirements. In this case, we improve the sum of uplink and downlink data rates of the entire system by optimizing user association, switching point selection, and subchannel allocation. The main contributions of this work are summarized below:
\begin{enumerate}
	\item We present a more realistic 5G heterogeneous wireless network architecture in which LTE BSs operating in the sub-6 GHz spectrum provide wide-area signal coverage, and pico BSs operating in the mmWave spectrum are densely deployed to enhance coverage in hotspots and edge areas. Both of them adopt orthogonal frequency division multiplexing \cite{fang2020energy} and time division duplex \cite{wang2005interference} technologies to establish wireless connections, and each pico BS is allowed to use an independent uplink and downlink switching point. In addition, we consider mobile user terminals, where different users may request diverse services, and the uplink and downlink data rate demands for the same service may be asymmetric.
	
	\item Unlike existing work that optimistically ignores mmWave interference or pessimistically treats it as omnidirectional propagation, we employ the more accurate propagation model of main and side beams: entities in the range of the main beam receive interference signals, while those in the range of the side beam ignore interference signals. Moreover, instead of crude mobility models such as random walk and random way point, we adopt the truncated Levy walk (TLW), which is empirically proven to be more consistent with human mobility patterns \cite{rhee2011levy}, to be the user mobility model.
	
	Based on these settings, we formulate a mathematical model for maximizing the uplink and downlink data rates of the system, which is integer nonlinear programming (INLP), under the constraints of the maximum capacity of the BS, single association of users, the feasibility of switching points, minimum uplink/downlink data rate requirement of users and communications continuity. Through a reduction from the Max-cut problem \cite{goemans1994879} to a subproblem of the original problem, we prove that the present problem is NP-hard.
	
	\item To solve the problem efficiently, we decouple it into the subproblems of user association, switching point selection, and subchannel allocation. The user association problem is modeled as an optimal matching problem of a bipartite with well-designed edge weights, and is solved using the Kuhn-Munkres algorithm. The switching point selection problem balances the uplink and downlink pseudo supply-demand ratios of the users served by the BS as much as possible. The subchannel allocation problem is proved to be equivalent to a generalized Max-cut problem of the interference graph, which is NP-hard. Hence, we transform it into a graph partitioning problem and solve it efficiently using spectral clustering at the expense of some optimality.
	
	\item In order to prove the effectiveness of the proposed algorithm, we perform extensive simulations to compare the performance of our algorithm with other baselines in the literature in terms of overall data rate, effective data rate, and the number of satisfied users. The results show that our algorithm has obvious advantages in these aspects.
\end{enumerate}

\subsection{Organization} \label{sec:organization}

The remainder of the paper is organized as follows. \secref{sec:system architecture} elaborates the system architecture and formulates the mathematical model. \secref{sec:algorithm design} analyzes the intractability of the problem and proposes an algorithm to solve it. \secref{sec:numerical simulations} presents the simulation results and compares the performance of different algorithms. \secref{sec:conclusion and future works} concludes the paper and discusses future works.

\section{System Architecture} \label{sec:system architecture}

In this section, we elaborate in detail on user association and channel allocation issues in heterogeneous wireless networks, explaining the roles of different entities in the system and modeling the optimization problem.

\begin{figure}[!th]
	\centering
	\includegraphics[scale=0.4]{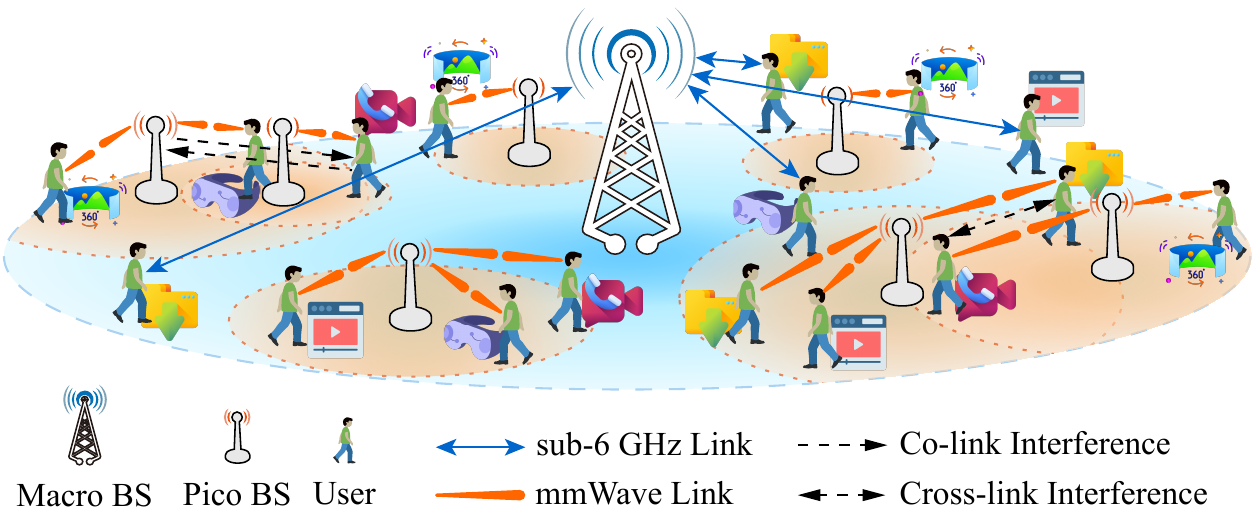}
	\caption{Network structure.}
	\label{fig:scenario}
\end{figure}

\subsection{Network Structure}

As shown in \figref{fig:scenario}, this work considers HetNets where macro and pico BSs coexist to provide wireless communication to various mobile users.

\begin{itemize}
	\item \textbf{Macro BS (MBS).} They are inherited from the LTE network and operate in the sub-6 GHz frequency band with omnidirectional antennas, mainly to provide wide area wireless signal coverage in 5G and beyond networks.
	
	\item \textbf{Pico BS (PBS).} This type of BS operates in the millimeter wave frequency band, thus, providing a larger communication frequency band compared to an MBS. They are equipped with directional antennas, which, combined with the directional propagation characteristics of millimeter waves, can achieve better antenna gain.
	
	However, the relatively smaller size, transmission power, and greater propagation loss of millimeter waves make the coverage range of PBSs much smaller than that of MBSs. Therefore, they are often densely deployed within macro cells as a supplement to the coverage of MBSs to enhance the communication capacity of the entire wireless network.
	
	\item \textbf{Mobile User Equipment (MUE).} They are mobile end devices, such as smartphones and tablets, that people use to access wireless network services. Due to the consideration of portability and battery life, they are also usually designed to be of small size and low power consumption, which is similar to PBSs.
	
	The physical laws require that the length of the antenna that can effectively receive wireless signals needs to be positively correlated with the carrier wavelength of the signal \cite{balanis2016antenna}. Therefore, the size of the antenna that can normally receive millimeter wave signals will only be smaller. This means that the end devices can simultaneously equipped with two types of antennas \cite{alizadeh2020multi, semiari2016downlink}, omnidirectional and directional, to support connections with both MBS and PBS.
	
	In addition, different services require different data rates to offer satisfying quality, and the uplink and downlink data rate requirements for the same service may also be asymmetric \cite{luo2014downlink}.
\end{itemize}

We set both BSs in \figref{fig:scenario} to serve multiple users simultaneously via orthogonal frequency-division multiple access (OFDMA) and adopt time division duplex (TDD) to transmit uplink and downlink data for each associated MUE \cite{wang2005interference}. Since the MUE is equipped with two types of antennas, they can establish a connection with either a macro or a pico BS, but at the same time slot, they can only communicate with one BS.

Due to the disparity in transmit power between MBSs and MUEs, all MBSs are synchronized at the uplink and downlink switching points within each time slot to avoid cross-link interference \cite{pedersen2016flexible}. On the contrary, each PBS is allowed to determine its switching point independently, because it has a similar transmit power to the user, making it not prone to severe cross-link interference.

During the communication process, users may move at any time and will carry their position information, real-time velocity information, and uplink/downlink data rate requirements when initiating communication requests. If their uplink or downlink data rate requirements are no longer met due to fluctuations in the wireless link, the users will initiate a reassociation request in the next time slot. Otherwise, they will maintain the current connection and the BS will not actively replace or interrupt communication with such users.

\subsection{Problem Formulation}

Before elaborating on the model of the problem, we list the major notations in \tabref{tab:notations} for quick reference.

\begin{table}[!tb]
	\begin{center}
		\caption{Notations}
		\label{tab:notations}
		\begin{tabular}{ p{5em} p{22em} }
			\hline
			$\mathcal{U}, \mathcal{U}^t$ & set of MUEs and those ask for reassociation at time slot~$t$ \\
			$\mathcal{B}, \mathcal{B}_\text{m}, \mathcal{B}_\text{p}$ & set of BSs, MBSs and PBSs, respectively \\
			$\mathcal{C}_b, \mathcal{C}_\text{m}, \mathcal{C}_\text{p}$ & set of subchannels of BS~$b$, Macro and PBS \\
			$W_b, W_\text{m}, W_\text{p}$ & subchannel bandwidth of BS~$b$, Macro and PBS \\
			$b^t_u$ & the BS that serves MUE~$u$ in time slot~$t$ \\
			$\tau_t$ & subslot within time slot~$t$ \\
			$x_{b,u}^{c,t}$ & binary variable that takes 1 when BS~$b$ serves MUE~$u$ with subchannel~$c$ during time slot~$t$ \\
			$d_{b,u}^{t}$ & distance between BS~$b$ and MUE~$u$ during slot~$t$ \\
			$G_b, G_u$ & antenna directivity of BS~$b$/UE~$u$ 
			\\
			$g_{b,u}^{c,t}, g_{u,b}^{c,t}$ & directional gain of BS~$b$/UE~$u$ to MUE~$u$/BS~$b$ on subchannel~$c$ in time slot~$t$ \\
			$\theta_b, \theta_u$ & beam width of the directional antenna of BS~$b$/UE~$u$ \\
			$\vartheta_b, \vartheta_u$ & sector width of BS~$b$/UE~$u$ \\
			$\boldsymbol{\psi}_{b,u}^{t}, \boldsymbol{\psi}_{u,b}^{t}$ & direction of the boresight from BS~$b$/UE~$u$ to MUE~$u$/BS~$b$ during time slot~$t$ \\
			$L_{b,u}^{t}$ & path loss between BS~$b$ and MUE~$u$ in time slot~$t$ \\
			$n_\text{L}^{\text{LTE}}, n_\text{N}^{\text{LTE}}$ & path loss exponent of LOS/NLOS sub-6 GHz link \\
			$n_\text{L}^{\text{mmW}}, n_\text{N}^{\text{mmW}}$ & path loss exponent of LOS/NLOS mmWave link 
			\\
			$P_b, P_u$ & transmitting power of BS~$b$/UE~$u$ on each subchannel \\
			$\lambda_\text{o}$, $\mathbb{E}_\text{L}$ & density and average length of obstacles, respectively \\
			$T_{\text{r}}$ & residence time after a single Levy flight \\
			$T_{b,u}$ & duration for beam alignment between BS~$b$ and MUE~$u$ \\
			$T_{\text{f}}, T_{\text{p}}, T_\text{s}$ & duration for a single flight, pilot transmission and time slot, respectively 
			\\
			$\beta_\text{f}, \beta_\text{r}$ & exponent of the length of a Levy flight and the residence time between Levy flights, respectively \\
			$N_\text{s}$ & number of subslot within each time slot 
			\\
			$N^{\text{TDD}}_{t,b}$ & subslot after which BS~$b$ switches to downlink in slot~$t$ \\
			$\mathcal{N}^{\text{TDD}}$ & set of optional TDD switching points \\
			$\gamma^{c,\tau_t}_{b,u}, \gamma^{c,\tau_t}_{u,b}$ & SINR received by MUE~$u$/BS~$b$ on subchannel~$c$ during subslot~$\tau_t$ from BS~$b$/UE~$u$ \\
			$I^{c,\tau_t}_{b,u}, I^{c,\tau_t}_{u,b}$ & interference suffered by MUE~$u$/BS~$b$ when receives data from BS~$b$/UE~$u$ on subchannel~$c$ during subslot~$\tau_t$ \\
			$r^{c,t}_{b,u}, r^{c,t}_{u,b}$ & perceived downlink/uplink rate between BS~$b$ and MUE~$u$ on subchannel~$c$ during time slot~$t$ \\
			$R^{\text{DL}}_u, R^{\text{UL}}_u$ & minimum downlink/uplink rate requirement of MUE~$u$ \\
			\hline
		\end{tabular}
	\end{center}
	\vspace{-15pt}
\end{table}

\subsubsection{Mobility pattern}

We denote the set of all MUEs as $\mathcal{U}$ and specify those that ask for reassociation at time slot~$t$ as $\mathcal{U}^t$. Their mobility pattern follows TLW, which is shown to have some statistical similarity to human walks in \cite{rhee2011levy}. The distance traveled by an MUE in a single walk is called a flight. The length of each flight and pause time between successive flights can be calculated by:
\begin{equation}\label{eq:levy walk length}
	l = \frac{y}{\left| z \right|^{1 / \beta}} ,
\end{equation}
where $\beta$ is a parameter that regulates the distribution of flight lengths. $y$ and $z$ are random variables that follow $N\left(0, \sigma^2_y\right)$ and $N\left(0, \sigma^2_z\right)$, respectively \cite{yang2010nature}. $\sigma_y$ and $\sigma_z$ are defined in \eqref{eq:levy SD}, where $\Gamma\left(\cdot\right)$ is the Gamma Function.
\begin{equation}\label{eq:levy SD}
	\sigma_y = \left\{\frac{\Gamma\left(1 + \beta\right) \sin\left(\pi \beta / 2\right)}{\Gamma\left[\left(1 + \beta\right) / 2\right] \beta 2^{\left(\beta - 1\right) / 2}}\right\}^{1 / \beta}, \quad \sigma_z = 1 ,
\end{equation}
\begin{equation}\label{eq:levy walk time}
	T_\text{f} = kl^{1-\rho}.
\end{equation}

The time needed to finish a flight is believed to be related to the flight length and is formulated in \eqref{eq:levy walk time}, where $k$ and $\rho$ are parameters that take the values of 30.55 and 0.89 respectively when flight length is smaller than 500 m. Otherwise, the two take 0.76 and 0.28, respectively \cite{rhee2011levy}. It reflects a phenomenon that people will always prefer transportation to walking for longer journeys, leading to an increase in both travel time and velocity as the distance of the flight increases.


\subsubsection{Interference structure}

The network is treated as a time-slotted system, where the channel state is static within time slots, but changes across time slots. The set of all BSs is denoted as $\mathcal{B}$, consisting of two disjoint sets $\mathcal{B}_\text{m}$ and $\mathcal{B}_\text{p}$ for MBSs and PBSs, respectively. The two types of base stations operate at different carrier frequencies, denoted as $f_\text{LTE}$ and $f_\text{mmW}$, respectively. The optional subchannels of BS~$b$ are $\mathcal{C}_b$ and the bandwidth of each subchannel is $W_b$.

The transmit power of BS~$b$ and MUE~$u$ are denoted as $P_b$ and $P_u$, respectively. We define a binary variable $x^{c,t}_{b,u}$ that takes 1 when BS~$b$ serves MUE~$u$ with subchannel~$c$ during time slot~$t$, and otherwise 0. For each PBS~$b$ and MUE~$u$ that communicate with mmWave band, we use $G_b$ and $G_u$ to represent their antenna directivities, and $\theta_b$ and $\theta_u$ to denote the main beam width of the directional antennas, respectively. The communicating parties of a mmWave link are considered to align their respective directional antennas in their LOS direction using beam alignment \cite{shokri2015beam} and beam tracking \cite{guo2020joint} techniques.

Now, we formulate the antenna gain on subchannel~$c$ in time slot~$t$ from BS~$b$ to BS~$b^\prime$, denoted as $g^{c,t}_{b, b^\prime}$, MUE~$u$ to MUE~$u^\prime$, denoted as  $g^{c,t}_{u, u^\prime}$, BS~$b$ to MUE~$u$, denoted as $g^{c,t}_{b,u}$, and MUE~$u$ to BS~$b$, denoted as $g^{c,t}_{u,b}$, as follows:
\begin{equation}\label{eq:antenna gain b2b}
	g_{b,b^\prime}^{c,t} = \left\{
	\begin{aligned}
		G_b, & \quad \begin{split}
			\text{if } & b, \! b^\prime \! \in \! \mathcal{B}_\text{p} \text{ and } \exists u \! \in \! \mathcal{U} \text{ that lets } x^{c,t}_{b,u}\!=\!1 \\
			& \text{and } \measuredangle\left(\boldsymbol{\psi}_{b,u}^{t}, \boldsymbol{\psi}^{t}_{b,b^\prime}\right) < \theta_b / 2
		\end{split} \\
		1, & \quad \text{else-if } b \in \mathcal{B}_\text{m} \text{ and } b^\prime \in \mathcal{B}_\text{m} \\
		0, & \quad \text{otherwise}
	\end{aligned}
	\right. ,
\end{equation}
\begin{equation}\label{eq:antenna gain u2u}
g_{u,u^\prime}^{c,t} = \left\{
\begin{aligned}
	G_u, & \quad \begin{split}
		\text{if } & u, u^\prime \text{ are mmWave users} \text{ and} \\
		& \exists b \! \in \! \mathcal{B}_\text{p} \text{ that lets } x^{c,t}_{b,u}\!=\!1 \text{ and} \\
		& \measuredangle\left(\boldsymbol{\psi}^{t}_{u,b}, \boldsymbol{\psi}^{t}_{u,u^\prime}\right) < \theta_u / 2
	\end{split} \\
	1, & \quad \text{else-if } u, u^\prime \text{ are LTE users} \\
	0, & \quad \text{otherwise}
\end{aligned}
\right. , \hspace*{29pt}
\end{equation}
\begin{equation}\label{eq:antenna gain b2u}
g_{b,u}^{c,t} = \left\{
\begin{aligned}
	G_b, & \quad \begin{split}
		\text{if } & b \! \in \! \mathcal{B}_\text{p} \text{ and } \exists u^\prime \in \mathcal{U} \text{ that lets } x^{c,t}_{b,u^\prime}\!=\!1 \\
		& \text{and } \measuredangle\left(\boldsymbol{\psi}_{b,u}^{t}, \boldsymbol{\psi}^{t}_{b,u^\prime}\right) < \theta_b / 2
	\end{split} \\
	1, & \quad \text{else-if } b \in \mathcal{B}_\text{m} \text{ and } c \in \mathcal{C}_b \\
	0, & \quad \text{otherwise}
\end{aligned}
\right. ,
\end{equation}
\begin{equation}\label{eq:antenna gain u2b}
	g_{u,b}^{c,t} = \left\{
	\begin{aligned}
		G_u, & \quad \begin{split}
			\text{if } & b \! \in \! \mathcal{B}_\text{p} \text{ and } \exists b^\prime \! \in \! \mathcal{B}_\text{p} \text{ that lets } x^{c,t}_{b^\prime,u}\!=\!1 \\
			& \text{and } \measuredangle\left(\boldsymbol{\psi}^{t}_{u,b}, \boldsymbol{\psi}^{t}_{u,b^\prime}\right) < \theta_u / 2
		\end{split} \\
		1, & \quad \text{else-if } b \in \mathcal{B}_\text{m} \text{ and } c \in \mathcal{C}_b \\
		0, & \quad \text{otherwise}
	\end{aligned}
	\right. ,
\end{equation}
where $\boldsymbol{\psi}^{t}_{b,u}$ is the boresight direction from BS~$b$ to MUE~$u$ in time slot~$t$, and the operator $\measuredangle$ measures the angle of the given vectors. $g^{c,t}_{\cdot, \cdot}$ indicates that the antenna gain between entities using the sub-6 GHz band is always 1. However, for those using the mmWave band, the antenna gain depends on the antenna directivity: it is effective only if the entities are on the same subchannel and one is within the main beam coverage of the other; otherwise, it is 0, meaning the side beam signal is ignored.

Within a time slot~$t$, the period is further divided into $N_\text{s}$ subslots, which is precisely the granularity of the uplink and downlink switching points. The characteristics of uplink and downlink asynchronous switching of the PBSs make the definition of the communication rate during a slot rely on that down to each subslot~$\tau_t$. We use $r^{c,\tau_t}_{b,u}$ to represent the instantaneous downlink rate from BS~$b$ to MUE~$u$ with subchannel~$c$ in subslot~$\tau_t$ within slot~$t$. It is formulated in \eqref{eq:Shannon}, where $\gamma^{c, \tau_t}_{b,u}$ is the Signal to Interference plus Noise Ratio (SINR) of the downlink signal from BS~$b$ to MUE~$u$ at that time and is defined in \eqref{eq:SINR}. We omit the presentation of $\gamma^{c, \tau_t}_{u,b}$, $\gamma^{c, \tau_t}_{b,b^\prime}$ and $\gamma^{c, \tau_t}_{u,u^\prime}$, since they are in the same form as $\gamma^{c, \tau_t}_{b,u}$.
\begin{equation}\label{eq:Shannon}
	r^{c,\tau_t}_{b,u} = W_b \log_2\left(1 + \gamma^{c, \tau_t}_{b,u} \right) ,
\end{equation}
\begin{equation}\label{eq:SINR}
	\gamma^{c,\tau_t}_{b,u} = \frac{P_b g^{c,t}_{b,u} g^{c,t}_{u,b} L^{-1}_{t, b, u}}{I_{b,u}^{c,\tau_t} + W_b N_0} .
\end{equation}

$L_{t, b, u}$ and $I_{b,u}^{c,\tau_t}$ above indicate the path loss between BS~$b$ and MUE~$u$ at slot~$t$ and the instantaneous interference suffered by the signal from BS~$b$ to MUE~$u$ at subslot~$\tau_t$, respectively. $N_0$ is the power spectral density of additive Gaussian White noise.
\begin{equation}\label{eq:FSPL}
	L_{t, b, u} = \left(\frac{4 \pi d^{t}_{b,u} f_b}{\text{c}}\right)^n .
\end{equation}

We adopt the free-space path loss, which is depicted in \eqref{eq:FSPL}, to measure the strength of path loss, where $d^t_{b,u}$ is the distance between BS~$b$ and MUE~$u$ at slot~$t$, $f_b$ is the carrier frequency of BS~$b$, nonitalic $\rm{c}$ is the speed of light, and $n$ is called the Path Loss Exponent (PLE). The value of PLE is highly dependent on the carrier frequency and the line-of-sight state between the BS and MUE. Many studies give typical values of PLE for sub-6 GHz \cite{blackard1992path, nwelih2022method} and mmWave band \cite{rappaport2013millimeter} in Line-of-Sight (LOS) and Non-Line-of-Sight (NLOS) states by fitting measurements in real environments.

We follow the results of these studies and focus on two states of wireless channels, i.e., LOS and NLOS. The probability of the channel being in an LOS state is related to the distance between the communicating parties \cite{bai2012using} and is given as follows:
\begin{equation}
	P^{\text{LOS}}_{t, b, u} = \exp\left(-\frac{2 \lambda_{\text{o}} \mathbb{E}_\text{L} d^{t}_{b,u}}{\pi}\right), 
\end{equation}
where $\lambda_{\text{o}}$ and $\mathbb{E}_\text{L}$ represent the density and average length of obstacles in the area, respectively.

Without loss of generality, we assume that uplink transmissions always come first during each time slot, followed by downlink ones. To explicitly formulate the instantaneous interference $I_{b,u}^{c,\tau_t}$, we let $N^{\text{TDD}}_{t,b}$ to be the subslot after which BS~$b$ switches to downlink in time slot~$t$, and $\mathbb{I}^x_y$ to be an indicator that takes the value of 1 when $x \leq y$ and 0 otherwise, while $\overline{\mathbb{I}}^x_y$ does the opposite. Hence, $I_{b,u}^{c,\tau_t}$ can be calculated using \eqref{eq:interference suffered}.
\begin{equation}\label{eq:interference suffered}
	\resizebox{0.89\linewidth}{!}{$
	\begin{aligned}
		I_{b,u}^{c,\tau_t} \!\! = \!\! \sum_{ \substack{b^\prime \neq b, 
			\\
			u^\prime \neq u} } \!\! x^{c,t}_{b^\prime\!,u^\prime} \!\! \left( \!\! \mathbb{I}^{\tau_t}_{N^{\text{TDD}}_{t,b^\prime}} \frac{ P_{u^\prime} g^{c,t}_{u^\prime\!, u} g^{c,t}_{u,u^\prime}}{L_{t, u^\prime\!, u}} \! + \!
		\overline{\mathbb{I}}^{\tau_t}_{N^{\text{TDD}}_{t,b^\prime}} \frac{ P_{b^\prime} g^{c,t}_{b^\prime\!, u} g^{c,t}_{u,b^\prime}}{L_{t, b^\prime\!, u}} \!\! \right)
	\end{aligned}\!.$}
\end{equation}

Based on the instantaneous uplink and downlink data rate, i.e., $r^{c,\tau_t}_{u,b}$ and $r^{c,\tau_t}_{b,u}$, of a subslot~$\tau_t$, the perceived data rate \cite{andrews2014overview} of the whole time slot~$t$, denoted as $r^{c,t}_{u,b}$ and $r^{c,t}_{b,u}$ for uplink and downlink, respectively, is formulated as follows:
\begin{equation}
	r^{c,t}_{u,b} = \left(1 - \frac{T_{b,u}}{T_\text{s}}\right) \frac{1}{N_\text{s}} \sum_{\tau_t = 1}^{N^\text{TDD}_{t,b}}{r^{c,\tau_t}_{u,b}} ,
\end{equation}
\begin{equation}
	r^{c,t}_{b,u} = \left(1 - \frac{T_{b,u}}{T_\text{s}}\right) \frac{1}{N_\text{s}} \sum_{\tau_t = N_{t,b}^{\text{TDD}} + 1}^{N_\text{s}}{r^{c,\tau_t}_{b,u}} ,
\end{equation}
where $T_\text{s}$ is the length of the time slot and $T_{b,u}$ is the time taken by beam alignment between BS~$b$ and MUE~$u$, which can be measured with \eqref{eq:beam alignment}. $\vartheta_b$ and $\vartheta_u$ are search sector width of BS~$b$ and MUE~$u$, respectively. The operator $\lceil \cdot \rceil$ indicates the ceiling function and $T_\text{p}$ is the pilot transmission time \cite{shokri2015beam}.
\begin{equation}\label{eq:beam alignment}
	T_{b,u} = \left\{\begin{aligned}
		0, \qquad\qquad & \text{if } b \in \mathcal{B}_\text{m} \\
		\bigg\lceil \frac{\vartheta_b}{\theta_b} \bigg\rceil \bigg\lceil \frac{\vartheta_u}{\theta_u} \bigg\rceil T_{\text{p}}, \quad & \text{if } b \in \mathcal{B}_\text{p}
	\end{aligned}\right. .
\end{equation}

\subsubsection{Problem model}

In this work, we aim to optimizing the connections between BSs and MUEs and adjust the up/down switching points of BSs to improve the communication data rates of mobile users in wireless HetNets. With the notations elaborated before, the mathematical formulation of the optimization problem in time slot~$t$ can be written as (P1).
\begin{align}
	(P1) \quad \underset{\boldsymbol{x}^t, \boldsymbol{N}^{\text{TDD}}_t}{\text{maximize}} \quad \sum_{u \in \mathcal{U}}{\sum_{b \in \mathcal{B}}{\sum_{c \in \mathcal{C}_b}{x^{c,t}_{b,u} \left(r^{c,t}_{u,b} + r^{c,t}_{b,u}\right)}}} \label{eq:sum data rate}\\
	\text{s.t.} \quad \sum_{u \in \mathcal{U}}{x^{c,t}_{b,u}} \leq 1, \quad \forall b \in \mathcal{B}, \forall c \in \mathcal{C}_b, \label{eq:serve at most one} \\
	\sum_{b \in \mathcal{B}}{\sum_{c \in \mathcal{C}_b}{x^{c,t}_{b,u}}} \leq 1, \quad \forall u \in \mathcal{U}, \label{eq:accept at most one} \\
	\sum_{b \in \mathcal{B}}{\sum_{c \in \mathcal{C}_b}{x^{c,t}_{b,u} r^{c,t}_{u,b}}} \geq R^{\text{UL}}_u \sum_{b \in \mathcal{B}}{\sum_{c \in \mathcal{C}_b}{x^{c,t}_{b,u}}}, \quad \forall u \in \mathcal{U}^t, \label{eq:UL demand} \\
	\sum_{b \in \mathcal{B}}{\sum_{c \in \mathcal{C}_b}{x^{c,t}_{b,u} r^{c,t}_{b,u}}} \geq R^{\text{DL}}_u \sum_{b \in \mathcal{B}}{\sum_{c \in \mathcal{C}_b}{x^{c,t}_{b,u}}}, \quad \forall u \in \mathcal{U}^t, \label{eq:DL demand} \\
	x^{c,t}_{b,u} = x^{c,t-1}_{b,u}, \quad \forall u \in \mathcal{U} \backslash \mathcal{U}^t, \forall b \in \mathcal{B}, \forall c \in \mathcal{C}_b, \label{eq:communication continuity} \\
	x^{c,t}_{b,u} \in \left\{0, 1\right\}, \quad \forall u \in \mathcal{U}, \forall b \in \mathcal{B}, \forall c \in \mathcal{C}_b, \label{eq:binary variable} \\
	N^{\text{TDD}}_{t,b} \in \mathcal{N}^{\text{TDD}}, \quad \forall b \in \mathcal{B}, \label{eq:optional switching point} \\
	N^{\text{TDD}}_{t,b} = N^{\text{TDD}}_{t,b^\prime}, \quad \forall b, b^\prime \in \mathcal{B}_\text{m}. \label{eq:synchronized MBS}
\end{align}

The objective in \eqref{eq:sum data rate} corresponds to the overall equivalent data rate obtained by all MUEs at time slot~$t$, including both uplink and downlink. $\boldsymbol{x}^t$ and $\boldsymbol{N}^{\text{TDD}}_t$ are decision vectors related to association and up/down switching, respectively. Constraint \eqref{eq:serve at most one} means that each subchannel of any BS can only serve at most one MUE during each time slot. Similarly, constraint \eqref{eq:accept at most one} means that each MUE can establish at most one connection with BSs through a single subchannel in each time slot. $R^{\text{UL}}_u$ and $R^{\text{DL}}_u$ in \eqref{eq:UL demand} and \eqref{eq:DL demand} are uplink and downlink rate demand of MUE~$u$, respectively. Hence, constraints \eqref{eq:UL demand} and \eqref{eq:DL demand} indicate that the equivalent data rates of MUEs should not be less than their demands in both directions. It is worth noting that the network side does not know the wireless channel states in advance, meaning that relying solely on the decision-maker cannot guarantee the results meet the two constraints. Therefore, we delegate this task to the user side. Specifically, the user always accepts the association decisions conveyed from the network side, but if the actual data rates obtained do not meet its requirements, the user will abort the data transmission and send a reconnection request at the beginning of the next time slot. Constraint \eqref{eq:communication continuity} indicates that the network cannot alter the serving BS and subchannel of an MUE that not ask for reassociation at the current time slot. Constraint \eqref{eq:binary variable} is an ordinary binary constraint. Constraint \eqref{eq:optional switching point} requires that the switching point of each BS can only drawn from the optional set, denoted as $\mathcal{N}^{\text{TDD}} = \left\{1, \dots, N_\text{s}-1\right\}$. Constraint \eqref{eq:synchronized MBS} means that the switching points of all MBSs should be kept consistent to avoid unacceptable cross-link interference.

The decision variables in (P1) are discrete variables and multiplication and division operations between them occur in the objective as well as in some constraints. These put the problem in the category of INLP, which is usually NP-hard and cannot be solved optimally in polynomial time. Indeed, we will prove that even a subproblem of (P1) remains NP-hard in the next section. Therefore, our basic principle for solving this problem is to search for a near-optimal solution quickly instead.

\section{Algorithm Design} \label{sec:algorithm design}

Considering the intractability of (P1), we solve it by decoupling it into several subproblems, namely, user association, switching point selection, and subchannel allocation, according to the connection establishment process. In this section, we describe the purpose and importance of these subproblems, analyze their characteristics, and propose the corresponding solutions.

\subsection{User Association}\label{sec:user association}

Finding the best BS is the primary and most important step in delivering wireless services to users. The quality of the association directly affects the upper limit of the data rate available to users and also determines the complexity of the interfering relationships throughout the network.

To formally define the user association problem, we use $x^t_{b,u}$, with a slight abuse, to represent whether MUE~$u$ is associated with BS~$b$ at time slot~$t$ or not, and introduce $w^t_{b,u}$ to denote the connection quality between BS~$b$ and MUE~$u$ at time slot~$t$, which is a physical quantity that is not affected by association variables. The goal of this subproblem is to obtain a near-optimal total connection quality under the constraints that each MUE can only connect to at most one BS, each BS can accommodate users up to the number of its subchannels and communication continuity is not compromised. We show the corresponding model in (P2), where the operator $\left|\cdot\right|$ measures the cardinality of the given set.
\begin{align}
	(P2) \quad \underset{\boldsymbol{x}^t}{\text{maximize}} \quad \sum_{u \in \mathcal{U}}{\sum_{b \in \mathcal{B}}{x^t_{b,u} w^t_{b,u}}} \label{eq:sum weights}\\
	\text{s.t.} \quad \sum_{u \in \mathcal{U}}{x^t_{b,u}} \leq \left|\mathcal{C}_b \right|, \quad \forall b \in \mathcal{B}, \label{eq:not exceeding subchannels} \\
	\sum_{b \in \mathcal{B}}{x^t_{b,u}} \leq 1, \quad \forall u \in \mathcal{U}, \label{eq:associate at most one} \\
	x^t_{b,u} = x^{t-1}_{b,u}, \quad \forall u \in \mathcal{U} \backslash \mathcal{U}^t, \forall b \in \mathcal{B}, \label{eq:association continuity} \\
	x^t_{b,u} \in \left\{0, 1\right\}, \quad \forall u \in \mathcal{U}, \forall b \in \mathcal{B}. \label{eq:binary association}
\end{align}

\begin{proposition}\label{prop:UA}
	The problem (P2) can be accurately and efficiently solved when $w^t_{b,u}$ is a constant value.
\end{proposition}
\begin{proof}
	We abstract all the BSs and MUEs in the wireless network into a bipartite, where each MUE corresponds to a vertex in the left vertex set and each BS corresponds to a series of vertices equal to the number of its subchannels in the right vertex set. The edge linking different sides of the graph possesses a weight that is exactly the connection quality of its endpoints.
	
	Based on the bipartite, we further remove the vertices that correspond to ongoing MUEs and one of their neighbor vertices that belong to the actual serving BS for each. The edges related to those vertices are also eliminated.
	
	\begin{align}
		\underset{\boldsymbol{x}^t \backslash \left\{x^t_{b,u} | u \in \mathcal{U} \backslash \mathcal{U}^t \right\}}{\text{maximize}} \quad \sum_{u \in \mathcal{U}^t}{\sum_{b \in \mathcal{B}}{x^t_{b,u} w^t_{b,u}}} \label{eq:sum weights of new MUEs} \\
		\text{s.t.} \quad \sum_{u \in \mathcal{U}}{x^t_{b,u}} \leq \left|\mathcal{C}^t_b \right|, \quad \forall b \in \mathcal{B}, \label{eq:not exceeding idle subchannels} \\
		\sum_{b \in \mathcal{B}}{x^t_{b,u}} \leq 1, \quad \forall u \in \mathcal{U}^t, \\
		x^t_{b,u} \in \left\{0, 1\right\}, \quad \forall u \in \mathcal{U}^t, \forall b \in \mathcal{B}. \label{eq:binary association for new MUEs}
	\end{align}
	
	With the transformation illustrated above, the optimal solution of (P2) that correspond to ongoing MUEs is determined. Meanwhile, those for other MUEs can be addressed from \eqref{eq:sum weights of new MUEs} to \eqref{eq:binary association for new MUEs}, where $\mathcal{C}^t_b$ is the set of idle subchannels of BS~$b$ at time slot~$t$. It is easy to see that the problem formulated in \eqref{eq:sum weights of new MUEs} to \eqref{eq:binary association for new MUEs} is equivalent to finding the optimal matching in the final bipartite we constructed, which is known to be solvable by Kuhn-Munkres algorithm \cite{munkres1957algorithms} in polynomial time.
	
	The transformation from (P2) to the bipartite can be done in polynomial time and vice versa. Therefore, we conclude that (P2) can be accurately solved in polynomial time.
\end{proof}

The proof procedure of \propref{prop:UA} has given a way to solve (P2), so the remaining question is how to define the connection quality $w^t_{b,u}$ so that it can better evaluate the BS's ability to serve the user in environments with wireless channel uncertainty, user mobility, and asymmetric data rate demands. We therefore define $w^t_{b,u}$ to be the smaller square root of the uplink and downlink pseudo-supply-demand ratios, which is formulated as:
\begin{equation}\label{eq:connection quality}
	w_{b,u}^t = \min \left\{ \sqrt{\frac{\tilde{r}^t_{u,b}}{R^{\text{UL}}_u}} \;, \quad \sqrt{\frac{\tilde{r}^t_{b,u}}{R^{\text{DL}}_u}} \right\},
\end{equation}
where $\tilde{r}^t_{u,b}$ and $\tilde{r}^t_{b,u}$ are pseudo data rates of uplink and downlink between BS~$b$ and MUE~$u$, respectively. We use the word "pseudo" here because $\tilde{r}^t_{u,b}$ and $\tilde{r}^t_{b,u}$ refer to the median of the estimated user data rates (rather than actual rates) averaged across all channel states, both at the current location and at the next location inferred from the velocity. The root-seeking operation is primarily intended to curb severe oversupply and schedule wireless resources more reasonably and fairly among MUEs \cite{ye2013user, liu2016user}. While the $\min$ operation is taken because the smaller one of the uplink and downlink supply-demand ratios is the real determinant of the user's communication experience.
\begin{equation}\label{eq:pseudo UL rate}
	\begin{aligned}
		\tilde{r}^t_{u,b} = P^{\text{LOS}}_{t, b, u} \frac{W_b}{2} \left[
		\log_2 \left(1 \! + \! \frac{P_u G_b G_u \left(L^{\text{LOS}}_{t, b, u}\right)^{-1}}{\tilde{I}_{u,b}^{t} + W_b N_0}\right) \right. \! + \! \\
		\left. \log_2 \left(1 + \frac{P_u G_b G_u \left(L^{\text{LOS}}_{t+1, b, u}\right)^{-1}}{\tilde{I}_{u,b}^{t} + W_b N_0}\right) \right] \left(1 - \frac{T_{b,u}}{T_\text{s}}\right) + \\
		P^{\text{NLOS}}_{t, b, u} \frac{W_b}{2} \left[
		\log_2 \left(1 + \frac{P_u G_b G_u \left(L^{\text{NLOS}}_{t, b, u}\right)^{-1}}{\tilde{I}_{u,b}^{t} + W_b N_0}\right) \right. + \\
		\left. \log_2 \left(1 + \frac{P_u G_b G_u \left(L^{\text{NLOS}}_{t+1, b, u}\right)^{-1}}{\tilde{I}_{u,b}^{t} + W_b N_0}\right) \right] \left(1 - \frac{T_{b,u}}{T_\text{s}}\right),
	\end{aligned}
\end{equation}
\begin{equation}\label{eq:pseudo DL rate}
	\begin{aligned}
		\tilde{r}^t_{b,u} = P^{\text{LOS}}_{t, b, u} \frac{W_b}{2} \left[
		\log_2 \left(1 \! + \! \frac{P_b G_b G_u \left(L^{\text{LOS}}_{t, b, u}\right)^{-1}}{\tilde{I}_{b,u}^{t} + W_b N_0}\right) \right. \! + \! \\
		\left. \log_2 \left(1 + \frac{P_b G_b G_u \left(L^{\text{LOS}}_{t+1, b, u}\right)^{-1}}{\tilde{I}_{b,u}^{t} + W_b N_0}\right) \right] \left(1 - \frac{T_{b,u}}{T_\text{s}}\right) + \\
		P^{\text{NLOS}}_{t, b, u} \frac{W_b}{2} \left[
		\log_2 \left(1 + \frac{P_b G_b G_u \left(L^{\text{NLOS}}_{t, b, u}\right)^{-1}}{\tilde{I}_{b,u}^{t} + W_b N_0}\right) \right. + \\
		\left. \log_2 \left(1 + \frac{P_b G_b G_u \left(L^{\text{NLOS}}_{t+1, b, u}\right)^{-1}}{\tilde{I}_{b,u}^{t} + W_b N_0}\right) \right] \left(1 - \frac{T_{b,u}}{T_\text{s}}\right).
	\end{aligned}
\end{equation}

Using pseudo data rate, rather than RSRP or RSSI, can alleviate the association imbalance caused by the disparity in transmission power between MBS and PBS. In addition, incorporating the information at the next location facilitates matching the user with the BS whose location better matches the direction of user movement, when a series of similar BSs are capable of handling. This design helps establish more persistent communication connections for users. The formulas of $\tilde{r}^t_{u,b}$ and $\tilde{r}^t_{b,u}$ are shown in \eqref{eq:pseudo UL rate} and \eqref{eq:pseudo DL rate}, where $\tilde{I}_{u,b}^{t}$ and $\tilde{I}_{b,u}^{t}$ are estimated interference experienced by uplink and downlink transmissions, respectively.

Depending on the value of the interference term in \eqref{eq:pseudo UL rate} and \eqref{eq:pseudo DL rate}, we adopt two schemes to get the pseudo data rate:

\begin{itemize}
	\item \textbf{SNR-based.} In this scheme, the interference term is forced to zero. In other words, we infer the connection quality based on the estimated signal-to-noise ratio.
	
	It should be emphasized that the adoption of SNR does not imply that our algorithm completely ignores interference similar to some previous work, but rather because the interference can be effectively reduced or even avoided through careful scheduling under the assistance of the OFDMA technique and the mmWave directivity. We will elaborate on these contents in \secref{sec:subchannel allocation}.
	
	\item \textbf{SINR-based.} Potential interference is defined as the sum of the average interference from all potential interferers. The interferers of a link, denoted as $\mathcal{O}^t_{u,b}$ for uplink and $\mathcal{O}^t_{b,u}$ for downlink, are BSs and MUEs that work on the same subchannel as the link and locate within the main beam range of the receiver. Since the uplink and downlink switching points of MBSs are synchronized, the interferers of the sub-6 GHz uplink are only MUEs and those of the sub-6 GHz downlink are only BSs.
	
	To unify the expression of the potential interference, we set the main beam width of the omnidirectional antenna to 360$^\circ$. Thus, the potential interference on uplink and downlink are formulated in \eqref{eq:potential UL interference suffered} and \eqref{eq:potential DL interference suffered}. Part of the coefficient is the probability that the interferer uses the same subchannel as the interfered entity and that the interfered one is also in the range of the interferer's main beam.
\end{itemize}
\begin{equation}\label{eq:potential UL interference suffered}
	\begin{aligned}
		\tilde{I}_{u,b}^{t} = \sum_{u^\prime \in \mathcal{O}^t_{u,b} \cap \mathcal{U}}{\frac{\theta_{u^\prime} P_{u^\prime} G_{u^\prime} G_b}{360 \cdot \left|\mathcal{C}_b\right|} \left(\frac{P^{\text{LOS}}_{t, u^\prime, b}}{L^{\text{LOS}}_{t, u^\prime, b}} + \frac{P^{\text{NLOS}}_{t, u^\prime, b}}{L^{\text{NLOS}}_{t, u^\prime, b}} \right)} +  \\
		\sum_{b^\prime \in \mathcal{O}^t_{u,b} \cap \mathcal{B}}{\frac{\theta_{b^\prime} P_{b^\prime} G_{b^\prime} G_b}{360 \cdot \left|\mathcal{C}_b\right|} \left(\frac{P^{\text{LOS}}_{t, b^\prime, b}}{L^{\text{LOS}}_{t, b^\prime, b}} + \frac{P^{\text{NLOS}}_{t, b^\prime, b}}{L^{\text{NLOS}}_{t, b^\prime, b}} \right)}, \;\;\;\;
	\end{aligned}
\end{equation}
\begin{equation}\label{eq:potential DL interference suffered}
	\begin{aligned}
		\tilde{I}_{b,u}^{t} = \sum_{u^\prime \in \mathcal{O}^t_{b,u} \cap \mathcal{U}}{\frac{\theta_{u^\prime} P_{u^\prime} G_{u^\prime} G_u}{360 \cdot \left|\mathcal{C}_b\right|} \left(\frac{P^{\text{LOS}}_{t, u^\prime, u}}{L^{\text{LOS}}_{t, u^\prime, u}} + \frac{P^{\text{NLOS}}_{t, u^\prime, u}}{L^{\text{NLOS}}_{t, u^\prime, u}} \right)} \! + \!  \\
		\sum_{b^\prime \in \mathcal{O}^t_{b,u} \cap \mathcal{B}}{\frac{\theta_{b^\prime} P_{b^\prime} G_{b^\prime} G_u}{360 \cdot \left|\mathcal{C}_b\right|} \left(\frac{P^{\text{LOS}}_{t, b^\prime, u}}{L^{\text{LOS}}_{t, b^\prime, u}} + \frac{P^{\text{NLOS}}_{t, b^\prime, u}}{L^{\text{NLOS}}_{t, b^\prime, u}} \right)}. \;\;\;\;
	\end{aligned}
\end{equation}

The SNR-based pseudo data rate provides an optimistic estimate of the actual connection quality and must be integrated with the subsequent subchannel allocation scheme to avoid excessively optimistic estimates. While the SINR-based pseudo data rate hopes to portray the actual connection quality more accurately and is a more conservative approach.

We name the algorithm proposed here to solve the user association subproblem the Optimal Matching based User Association (OMUA). The complete procedure of OMUA is summarized in \algref{alg:OMUA}.

\begin{algorithm}[!tb]
	\caption{OMUA} \label{alg:OMUA}
	\SetKwInOut{Input}{Input}
	\SetKwInOut{Output}{Output}
	\Input{1) Set $\mathcal{B}$ of BSs;\newline
		2) Set $\mathcal{U}^t$ of MUEs that ask for association;\newline
		3) Position information, velocity information and data rate demands of each MUE in $\mathcal{U}$.}
	\Output{The association result $x^t_{b,u}$ for each MUE in $\mathcal{U}$.}
	
	Maintain the association status of MUEs in $\mathcal{U} \backslash \mathcal{U}^t$.
	
	Construct a bipartite, where each MUE in $\mathcal{U}^t$ maps to a vertex in the left vertex set, each BS in $\mathcal{B}$ maps to several vertices equal to the number of its subchannels minus the number of its ongoing MUEs in the right vertex set.
	
	Establish edges between any two vertices in different parts of the bipartite, and assign edge weights according to \eqref{eq:connection quality}. \label{line:measure edge weight}
	
	If this bipartite is not balanced, auxiliary vertices and edges are added to make it balanced, and the weights of added edges are set to 0.
	
	Invoke the Kuhn-Munkres algorithm to solve the optimal matching of the balanced bipartite. \label{line:KM}
	
	Set the BS that each MUE in $\mathcal{U}^t$ will be associated with based on the matching results.
	
	\textbf{return} the association decision $x^t_{b,u}$ for each MUE in $\mathcal{U}$.
\end{algorithm}

\subsection{Switching Point Selection}\label{sec:switching point selection}

The purpose of selecting a TDD switching point for a BS is to schedule communication resources of the BS more adequately in the time dimension, which can play a more significant role than a fixed scheme when the users of the BS have similar uplink and downlink communication requirements. The ideal switching point of a BS should make the uplink and downlink pseudo-supply-demand ratios of all its users most approximate. However, such a switching point cannot be efficiently solved most of the time.

Fortunately, for a certain MUE~$u$ associated with BS~$b$, the optimal switching point $N^\text{TDD}_{t,b,u}$ is the point where the uplink and downlink pseudo-supply-demand ratios of MUE~$u$ are equal, and it can be easily solved using \eqref{eq:individual optimal switching point}, or \eqref{eq:individual optimal switching point transformed}.
\begin{equation}\label{eq:individual optimal switching point}
	\frac{N^\text{TDD}_{t,b,u}}{N_{\text{s}}} \frac{\tilde{r}^t_{u, b}}{R^{\text{UL}}_u} = \left(1 - \frac{N^\text{TDD}_{t,b,u}}{N_{\text{s}}} \right) \frac{\tilde{r}^t_{b, u}}{R^{\text{DL}}_u},
\end{equation}
\begin{equation}\label{eq:individual optimal switching point transformed}
	N^\text{TDD}_{t,b,u} = \frac{R^{\text{UL}}_u \tilde{r}^t_{b, u}}{R^{\text{UL}}_u \tilde{r}^t_{b, u} + R^{\text{DL}}_u \tilde{r}^t_{u, b}} N_{\text{s}}.
\end{equation}

Taking into account all the MUEs associated with the same BSs, we choose the average of their optimal switching points as the switching point for that BSs. It should be noted that all MBSs need to synchronize their switching points, so their switching points are set to the average of the optimal switching points of all MUEs using MBSs. Therefore, the switching point of BS~$b$ at time slot~$t$ in the form of fractions can be formulated as:
\begin{equation}\label{eq:selected switching point}
	N^\text{TDD}_{t,b} = \left\{\begin{aligned}
		& \frac{\sum_{u \in \cup_{b^\prime \in \mathcal{B}_\text{m}} \mathcal{U}^t_{b^\prime}} N^\text{TDD}_{t,b^\prime,u}}{\left| \cup_{b^\prime \in \mathcal{B}_\text{m}} \mathcal{U}^t_{b^\prime} \right|}, \quad & \text{if } b \in \mathcal{B}_\text{m}
		\\
		& \quad\quad \frac{\sum_{u \in \mathcal{U}^t_b} N^\text{TDD}_{t,b,u}}{\left| \mathcal{U}^t_b \right|}, \quad & \text{if } b \in \mathcal{B}_\text{p}
	\end{aligned} \right.,
\end{equation}
where $\mathcal{U}^t_b$ denotes the set of MUEs that are served by BS~$b$ after the user association process at time slot~$t$. Recall that $N^\text{TDD}_{t,b}$ needs to satisfy constraint \eqref{eq:optional switching point}, so the result of \eqref{eq:selected switching point} needs to be rounded up and then clipped by the upper and lower bounds of $\mathcal{N}^\text{TDD}$ to become the final switching point of a BS.

\subsection{Subchannel Allocation}\label{sec:subchannel allocation}

After completing the user association and BS switching point selection, we manage to determine the specific sub-channels to be used for each connection. This is a necessary operation for establishing a wireless communication link and also an effective means of suppressing co-channel interference between MUEs.

The purpose of this subproblem is to set up communication subchannels for all MUEs that are associated with a BS. It aims to minimize the co-channel interference across the system, ensuring that existing subchannels utilized by ongoing MUEs remain unchanged, and that MUEs under the same BS operate on distinct subchannels. We assert and prove that the model of this problem can be formulated as:
\begin{align}
	(P3) \;\; \underset{\boldsymbol{c}^t}{\text{minimize}} \; \frac{1}{2} \sum_{u, u^\prime \in \mathcal{U}}{\max \left\{0, 1 -\left| c^t_u - c^t_{u^\prime} \right| \right\} i^t_{u,u^\prime}} \\
	\text{s.t.} \quad c^t_u \in \mathcal{C}^t_u, \quad \forall u \in \mathcal{U},
\end{align}
where $c^t_u$ denotes the subchannel assigned to MUE~$u$ at time slot~$t$ and $\mathcal{C}^t_u$ denotes the set of optional subchannels of MUE~$u$. $i^t_{u,u^\prime}$ indicates the interference power between the connections of MUE~$u$ and MUE~$u^\prime$.

If MUE~$u$ is not matched with the BS, $\mathcal{C}^t_u$ is an empty set. Otherwise, if MUE~$u$ is an ongoing user, $\mathcal{C}^t_u$ contains only the subchannel it used in the last time slot. Otherwise, $\mathcal{C}^t_u$ is equivalent to the set of subchannels of the BS to which MUE~$u$ is connected. The self-interference power of a connection is treated as zero, and that between all connections belonging to the same BS or belonging to ongoing MUEs using different subchannels are treated as infinite. Meanwhile, the interference power between the rest connections refers to the largest one of the average interference that one MUE may suffer from the other MUE under the assumption that both of them use the same subchannel, which is given as follows:
\begin{equation}\label{eq:interference between connections}
	\resizebox{0.89\linewidth}{!}{$
	\begin{aligned}
		i_{u, u^\prime}^t = N^{-1}_\text{s}\max \left\{ \left( N_\text{s} - \max \left\{ N^\text{TDD}_{t,b_u}, N^\text{TDD}_{t,b_{u^\prime}} \right\} \right) \overline{i}^t_{b_u, u^\prime}, \right.
		\\
		\left( N_\text{s} - \max \left\{ N^\text{TDD}_{t,b_u}, N^\text{TDD}_{t,b_{u^\prime}} \right\}\right) \overline{i}^t_{b_{u^\prime}, u},
		\\
		\min \left\{ N^\text{TDD}_{t,b_u}, N^\text{TDD}_{t,b_{u^\prime}} \right\} \overline{i}^t_{u, b_{u^\prime}},
		\min \left\{ N^\text{TDD}_{t,b_u}, N^\text{TDD}_{t,b_{u^\prime}} \right\} \overline{i}^t_{u^\prime, b_u},
		\\
		\max \! \left\{ \! 0, \! N^{\text{TDD}}_{t,b_{u^\prime}} \! - \! N^{\text{TDD}}_{t,b_u} \! \right\} \! \overline{i}^t_{b_u, b_{u^\prime}},
		\max \! \left\{ \! 0, \! N^{\text{TDD}}_{t,b_u} \! - \! N^{\text{TDD}}_{t,b_{u^\prime}} \! \right\} \! \overline{i}^t_{b_{u^\prime} \! , b_u},
		\\
		\max \! \left\{ \! 0, \! N^{\text{TDD}}_{t,b_u} \! - \! N^{\text{TDD}}_{t,b_{u^\prime}} \! \right\} \! \overline{i}^t_{u, u^\prime},
		\left.
		\max \! \left\{ \! 0, \! N^{\text{TDD}}_{t,b_{u^\prime}} \! - \! N^{\text{TDD}}_{t,b_u} \! \right\} \! \overline{i}^t_{u^\prime \!, u}
		\! \right\},
	\end{aligned}$}
\end{equation}
where $b_u$ denotes the BS that serves MUE~$u$ and $\overline{i}^t_{x, y}$ denotes the average interference caused by object~$x$ to object~$y$. The first four $\overline{i}^t_{x, y}$ in \eqref{eq:interference between connections} correspond to the co-link interference, the remainings are the cross-link interference, which can be obtained from \eqref{eq:colink interferenc u2b} to \eqref{eq:crosslink interferenc u2u}.
\begin{equation}\label{eq:colink interferenc u2b}
	\overline{i}^t_{u, b_{u^\prime}} = P_u g^{c^t_u,t}_{u, b_{u^\prime}} g^{c^t_u,t}_{b_{u^\prime}, u} \left( \frac{ P^{\text{LOS}}_{t, u, b_{u^\prime}} }{L^{\text{LOS}}_{t, u, b_{u^\prime}}} + \frac{ P^{\text{NLOS}}_{t, u, b_{u^\prime}} }{L^{\text{NLOS}}_{t, u, b_{u^\prime}}} \right),
\end{equation}
\begin{equation}\label{eq:colink interferenc b2u}
	\overline{i}^t_{b_u, u^\prime} = P_{b_u} g^{c^t_u,t}_{b_u, u^\prime} g^{c^t_u,t}_{u^\prime, b_u} \left( \frac{ P^{\text{LOS}}_{t, b_u, u^\prime} }{L^{\text{LOS}}_{t, b_u, u^\prime}} + \frac{ P^{\text{NLOS}}_{t, b_u, u^\prime} }{L^{\text{NLOS}}_{t, b_u, u^\prime}} \right),
\end{equation}
\begin{equation}\label{eq:crosslink interferenc b2b}
	\overline{i}^t_{b_u, b_{u^\prime}} = P_{b_u} g^{c^t_u,t}_{b_u, b_{u^\prime}} g^{c^t_u,t}_{b_{u^\prime},b_u} \left( \frac{ P^{\text{LOS}}_{t, b_u, b_{u^\prime}} }{L^{\text{LOS}}_{t, b_u, b_{u^\prime}}} + \frac{ P^{\text{NLOS}}_{t, b_u, b_{u^\prime}} }{L^{\text{NLOS}}_{t, b_u, b_{u^\prime}}} \right),
\end{equation}
\begin{equation}\label{eq:crosslink interferenc u2u}
	\overline{i}^t_{u, u^\prime} = P_u g^{c^t_u,t}_{u, u^\prime} g^{c^t_u,t}_{u^\prime,u} \left( \frac{ P^{\text{LOS}}_{t, u, u^\prime} }{L^{\text{LOS}}_{t, u, u^\prime}} + \frac{ P^{\text{NLOS}}_{t, u, u^\prime} }{L^{\text{NLOS}}_{t, u, u^\prime}} \right).
\end{equation}

\begin{lemma}\label{lem:finite or infinite interference}
	The solutions of (P3) that assign different subchannels for MUEs of the same BS and for those ongoing MUEs using different subchannels previously always yield finite objective values. Other solutions produce infinite objective values.
\end{lemma}
\begin{proof}
	It is clear that the subchannel allocation scheme satisfying the above conditions will only retain $i^t_{u, u^\prime}$ calculated by \eqref{eq:interference between connections} in the optimization objective, which are all finite values. On the contrary, it will certainly retain infinite terms in the optimization objective.
\end{proof}
\begin{lemma}\label{lem:exist a solution produce finite interference}
	There is always a solution of (P3) that produces a finite objective value.
\end{lemma}
\begin{proof}
	Recall the procedure of constructing the bipartite in \algref{alg:OMUA}, it ensures that the number of MUEs ultimately associated with each BS does not exceed its total number of subchannels. Therefore, there must exist feasible allocation schemes that allow each MUE of the same BS to occupy different subchannels. At the same time, each BS can maintain the original subchannel for its (possibly existing) ongoing MUEs.
	
	These imply that there necessarily exists allocation schemes that satisfy the conditions mentioned in \lemref{lem:finite or infinite interference}. Therefore, these allocation schemes produce finite objective values.
\end{proof}
\begin{proposition}\label{prop:optimal solution produce what we want}
	The optimal solution to (P3) is exactly the allocation scheme pursued by the subchannel allocation problem.
\end{proposition}
\begin{proof}
	According to \lemref{lem:exist a solution produce finite interference}, the optimal solution of (P3) must yield a finite objective value. Combined with \lemref{lem:finite or infinite interference}, this optimal solution has to satisfy the conditions stated in \lemref{lem:finite or infinite interference}, which are also the constraints that need to be satisfied for the subchannel allocation problem.
	
	Notice that there may be more than one optimal solution to (P3), but all of them can ensure that all ongoing MUEs obtain original subchannels by simple collective subchannel swaps. The swap operation does not change the objective value and does not violate the conditions mentioned in \lemref{lem:finite or infinite interference}.
	
	In summary, the optimal solution of (P3) achieves the objective consistent with the subchannel allocation problem and satisfies all its constraints.
\end{proof}

We have shown in \propref{prop:optimal solution produce what we want} that solving (P3) accomplishes the task of subchannel assignment; unfortunately, (P3) is intractable. To prove this, we represent (P3) as (P4) through a simple transformation:
\begin{align}
	(P4) \;\; \underset{\boldsymbol{c}^t}{\text{maximize}} \; \frac{1}{2} \sum_{u, u^\prime \in \mathcal{U}}{\min \left\{1, \left| c^t_u - c^t_{u^\prime} \right| \right\} i^t_{u,u^\prime}} \\
	\text{s.t.} \quad c^t_u \in \mathcal{C}^t_u, \quad \forall u \in \mathcal{U}.
\end{align}

The problem given in (P4) can be regarded as a graph partitioning problem, i.e., dividing all vertices of a given graph $\mathcal{G}$ into a specified number of clusters such that the sum of edge weights between clusters is maximized. We formally define the graph partitioning problem as:
\begin{align}
	(P5) \;\; \underset{\mathcal{G}_1, \dots, \mathcal{G}_{\left| \mathcal{C} \right|}}{\text{maximize}} \quad \frac{1}{2} \sum_{c = 1}^{\left| \mathcal{C} \right|}{ W \left( \mathcal{G}_c, \mathcal{G}_{-c} \right) }
	\\
	\text{s.t.} \quad \mathcal{G}_c \cap \mathcal{G}_{c^\prime} = \emptyset, \;\; \forall c, c^\prime \in \mathcal{C},
	\\
	\mathcal{G}_1 \cup \cdots \cup \mathcal{G}_{\left| \mathcal{C} \right|} = \mathcal{G},
\end{align}
where $\mathcal{G}_c$ is the set of vertices in cluster~$c$, $\mathcal{G}_{-c}$ is the set of vertices that are not in cluster~$c$ and $\left|\mathcal{C}\right|$ is the number of clusters to be divided. $W \left(\mathcal{G}_1, \mathcal{G}_2 \right)$ measures the sum of edge weights between cluster~$\mathcal{G}_1$ and $\mathcal{G}_2$, as follows:
\begin{equation}\label{eq:sum of edge weights between two clusters}
	W \left(\mathcal{G}_1, \mathcal{G}_2 \right) = \sum_{u \in \mathcal{G}_1, v \in \mathcal{G}_2} w_{u,v}.
\end{equation}

Recall that MBSs and PBSs use different frequency bands and their respective MUEs can only use the corresponding subchannels without interfering with each other. Therefore, (P4) can be treated as two separate graph partitioning problems, one graph corresponding to MBS MUEs and the other corresponding to PBS MUEs: the set of subchannels is mapped to the set of clusters, the interference power between MUE connections is mapped to the edge weight, and ongoing MUEs originally using the same subchannel collapse into a single vertex.

\begin{proposition}
	The subchannel allocation problem in (P4) is NP-hard.
\end{proposition}
\begin{proof}
	We draw this conclusion by a reduction from the Max-cut problem, which aims to divide a graph into two parts such that the sum of the weights of the cut edges is maximized.
	
	Any algorithm that can find the optimal solution to (P5) can also solve the Max-cut problem: simply set the number of clusters to 2. In other words, the Max-cut problem is a special case of (P5), i.e., the Max-cut problem can be reduced to (P5).
	
	Since the Max-cut problem is a well-known NP-hard problem \cite{goemans1994879}, the problem shown in (P5), which we have proved at least as difficult to solve as the Max-cut problem, is also an NP-hard problem. Meanwhile, recall that (P5) is another form of (P4) from the perspective of graph partitioning, so (P4) is NP-hard.
\end{proof}

It is well known that the NP-hard problem cannot be solved accurately in polynomial time unless P = NP. To ensure the usability of the algorithm in wireless environments, we sacrifice the optimality of the solution in favor of an approximate one. Based on the graph perspective presented in (P5), we turn our attention to the clustering problem and find that Spectral Clustering \cite{von2007tutorial} solves a problem that is similar to (P5):
\begin{itemize}
	\item Spectral clustering also involves dividing the vertices in the graph into a specified number of clusters.
	
	\item Spectral clustering also deals with edge-weighted graphs, where edge weights measure the similarity of endpoints.
	
	\item Spectral clustering wants to classify vertices with higher similarity into the same cluster, so its objective is to minimize the sum of edge weights between clusters.
\end{itemize}

Spectral clustering appears to pursue the opposite goal of (P5), but in a physical sense, MUE connections that interfere with each other more intensely should be more dissimilar \cite{von2007tutorial}. Thus, the gap between (P5) and the problem handled by spectral clustering is a function that maps the intensity of interference to similarity, we call it the similarity function.

Before elaborating on the similarity function, we distinguish between two classes of MUEs:
\begin{itemize}
	\item \textbf{Conflicting MUEs.} They refer to MUEs that should never use the same subchannel, i.e., MUEs associated with the same BS, as well as ongoing MUEs that originally used different subchannels.
	
	\item \textbf{Interfering MUEs.} They refer to MUEs that generate the finite interference defined in \eqref{eq:interference between connections} when using the same subchannel.
\end{itemize}

We define the similarity between conflicting MUEs to be 0, and that between interfering MUEs to be the smaller ratio of the total interference suffered by one side from all interfering MUEs except and including the other side as follows:
\begin{equation}\label{eq:similarity between MUEs}
	s^t_{u,u^\prime} = \min \left\{ \frac{\sum_{v \neq u^\prime \in \mathcal{U}^\text{I}_{t,u}} i^t_{u,v}}{\sum_{v \in \mathcal{U}^\text{I}_{t,u}} i^t_{u,v}}, \frac{\sum_{v \neq u \in \mathcal{U}^\text{I}_{t,u^\prime}} i^t_{v,u^\prime}}{\sum_{v \in \mathcal{U}^\text{I}_{t,u^\prime}} i^t_{v,u^\prime}} \right\},
\end{equation}
where $\mathcal{U}^\text{I}_{t,u}$ denotes the set of MUEs that interfere with MUE~$u$ at time slot~$t$. Thus, the similarity function $\mathbb{S}\left(\cdot, \cdot\right)$ can be formulated in \eqref{eq:similarity function}. We must emphasize that the mapping here is not lossless, i.e., the problems before and after the mapping are not exactly equivalent, which is an unavoidable compromise due to the intractability of (P5).
\begin{equation}\label{eq:similarity function}
	\mathbb{S} \left(u, u^\prime \right) = \left\{\begin{aligned}
		0, & \quad \text{for conflicting MUEs}
		\\
		s^t_{u, u^\prime}, & \quad \text{for interfering MUEs}
	\end{aligned}\right..
\end{equation}

The definition of the similarity function ensures that the similarity between all MUEs is symmetric and lies within the interval $[0,1]$, which means that the similarity matrix $S^t$ consisting of $\mathbb{S} \left(u, u^\prime \right)$ is a symmetric matrix with non-negative elements, thus satisfying the requirements of spectral clustering. With the similarity matrix $S^t$, we define the Laplacian:
\begin{equation}\label{eq:Laplacian}
	\mathbb{L}^t = D^t - S^t,
\end{equation}
where $D^t$ is the diagonal matrix of the row sum of $S^t$. By feeding $\mathbb{L}^t$ to spectral clustering, we can obtain clustering results for all MUEs.

We have proved in \propref{prop:optimal solution produce what we want} that the optimal solution of (P3) is the best scheme for subchannel allocation. This conclusion readily extends to (P4) and (P5), given their equivalence at the optimal solution. However, the similarity function $\mathbb{S}\left(\cdot\right)$ introduces some gap between (P5) and the clustering problem, and spectral clustering is a method that sacrifices optimality for efficiency, which leads to the fact that its solution may not satisfy the constraints of subchannel allocation. Spectral clustering does not guarantee to classify conflicting MUEs into different clusters.

To solve the above problem, our basic idea is to adjust the clustering results by replacing the conflicting MUEs classified into the same cluster with new clusters until there are no conflicting MUEs in any of the clusters. Fortunately, according to \lemref{lem:exist a solution produce finite interference}, such a clustering scheme must exist and the adjustment operation can be done in polynomial time.

Specifically, we start checking from larger clusters, if there are conflicting ongoing MUEs, we find their respective target clusters, which are defined as the rest clusters where there are no other MUEs in conflict with them and the sum of the cut-edge weights, i.e., the interference power $i^t_{u, u^\prime}$, is the smallest. Then we compute the increment to the cut-edge weights brought by their switching to their respective target clusters, the conflicting ongoing MUE with the largest increment is selected to switch to its target cluster. This process is repeated until all clusters are free of conflicting ongoing MUEs. After that, we continue to check for conflicting but non-ongoing MUEs in the clusters. If such MUEs exist, we perform the same operation as above on them until all clusters are free of any conflicting MUEs.

Based on the refined clustering results, we correspond each cluster to a specific subchannel. Specifically, the clusters containing ongoing MUEs correspond to the subchannels originally used by the ongoing MUEs, while the other clusters simply correspond to any different remaining subchannels. MUEs in the cluster eventually obtain the subchannel corresponding to the cluster.

We name the algorithm devised to solve the subchannel allocation subproblem the Spectral Clustering Subchannel Allocation (SCSA) and summarize the procedure in \algref{alg:SCSA}.

\begin{algorithm}[!tb]
	\caption{SCSA} \label{alg:SCSA}
	\SetKwInOut{Input}{Input}
	\SetKwInOut{Output}{Output}
	\Input{1) Association results $x^t_{b,u}$ of MUEs;\newline
		2) TDD switching points $N^\text{TDD}_{t,b}$ of BSs.}
	\Output{Subchannel allocation results $x^{c,t}_{b,u}$ for MUEs.}
	Divide macro and pico MUEs into groups $\mathcal{U}_\text{m}$ and $\mathcal{U}_\text{p}$.
	
	\For{each group in  $\left\{\mathcal{U}_\text{m}, \mathcal{U}_\text{p}\right\}$}{
		Set the interference power $i^t_{u, u^\prime}$ between conflicting MUEs to $\infty$ and others to \eqref{eq:interference between connections}.
		
		Construct interference graph $\mathcal{G}$.
	
		Collapse the vertices in $\mathcal{G}$ corresponding to ongoing MUEs originally using the same subchannel into a single vertex. If multiple edges appear, merge edge weights.
		
		Convert the edge weights in $\mathcal{G}$ to similarity as \eqref{eq:similarity function}.
		
		\tcp{Spectral Clustering}
		Construct the Laplacian matrix $\mathbb{L}^t$.
		
		Set the number of clusters $\left|\mathcal{C}\right|$ to be the smaller one of the number of subchannels and the number of vertices $\left|\mathcal{G}\right|$.
		
		Compute the eigenvectors, denoted as $\left\{e_1, \dots, e_{\left|\mathcal{C}\right|}\right\}$, corresponding to the first $\left|\mathcal{C}\right|$ smallest eigenvalues of $\mathbb{L}^t$.
		
		Cluster the points $\left\{y_i | i = 1, \dots, \left|\mathcal{G}\right| \right\}$ in $\mathbb{R}^{\left|\mathcal{C}\right|}$, where $y_i$ is the $i$-th row of matrix $\left[e_1, \dots, e_{\left|\mathcal{C}\right|}\right]$, with the K-means algorithm into $\left|\mathcal{C}\right|$ clusters.
		
		\tcp{Subchannel Correction}
		\For{each cluster having conflicting ongoing MUEs}{
			For such MUEs, find their respective target clusters, where there are no other MUEs in conflict with them and the sum of the cut-edge weights $i^t_{u, u^\prime}$ is the smallest.
			
			Calculate for each MUE the increment of the cut-edge weights that switching to the target cluster brought.
			
			Move the MUE that yields highest increment to its target cluster.
		}
		
		\For{each cluster having other conflicting MUEs}{
			For such MUEs, find their target clusters and calculate the increment.
			
			Move the MUE that yields highest increment to its target cluster.
		}
		
		The clusters containing ongoing MUEs are assigned the original subchannels, the other clusters are assigned any different remaining subchannels.
	}

	\textbf{return} the subchannel $x^{c,t}_{b,u}$ assigned to each MUE.
\end{algorithm}

\subsection{Complete Algorithm}

We elaborate on the subproblems of (P1) and the schemes designed to solve them in \secref{sec:user association} to \secref{sec:subchannel allocation}. Integrating all the schemes yields the complete algorithm we propose for solving (P1), which is called OMSC.

OMSC first selects more suitable BSs for all MUEs based on optimal matching. Given the association results of the MUEs, the BS can then estimate the TDD switching point that balances the supply-demand ratio of uplink and downlink data rates. Moreover, the outcomes of the two successive schedulings clarify and fix the spatial coverage relationship between the signals of the MUEs, which reduces the uncertainty of interference management. On this basis, we carry out subchannel allocation optimization based on spectral clustering and try to suppress interference as much as possible using frequency division multiplexing. The detailed procedure of OMSC is shown in \algref{alg:OMSC}.

\begin{algorithm}[!tb]
	\caption{OMSC} \label{alg:OMSC}
	\SetKwInOut{Input}{Input}
	\SetKwInOut{Output}{Output}
	\Input{1) Set $\mathcal{B}$ of BSs;\newline
		2) Set $\mathcal{U}^t$ of MUEs that ask for association;\newline
		3) Position information, velocity information and data rate demands of each MUE in $\mathcal{U}$.}
	\Output{1) Decisions $x^{c,t}_{b,u}$ for each MUE;\newline
		2) TDD switching points $N^\text{TDD}_{t,b}$ for each BS.}
	
	\tcp{The association decision}
	Invoke \textit{OMUA} in \algref{alg:OMUA} given $\mathcal{B}$, $\mathcal{U}^t$, and information related to MUEs as inputs to get the association decisions $x^t_{b,u}$.
	
	\tcp{The TDD switching point selection}
	\For{each BS in $\mathcal{B}$}{
		Estimate the ideal switching point for each associated MUE according to \eqref{eq:individual optimal switching point transformed}.
		
		Compute the average value of the ideal switching point of these MUEs according to \eqref{eq:selected switching point}.
		
		Round the average value, clip it with the upper and lower bounds of the set $\mathcal{N}^\text{TDD}$, and adopt the result to be the actual switching point of the BS.
	}
	
	\tcp{The subchannel allocation}
	Invoke \textit{SCSA} in \algref{alg:SCSA} given $x^t_{b,u}$ and $N^\text{TDD}_{t,b}$ as inputs to get the subchannel assignments $x^{c,t}_{b,u}$ for each MUE.
	
	\textbf{return} decisions $x^{c,t}_{b,u}$ for each MUE and switching points $N^\text{TDD}_{t,b}$ for each BS.
\end{algorithm}

\subsection{Complexity Analysis} \label{sec:complexity analysis}

In this subsection, we will analyze each stage of OMSC to derive the complexity of the whole algorithm.

The main time overhead of OMUA stems from the computation of the edge weights of the bipartite and the Kuhn-Munkres algorithm to find the optimal matching, which corresponds to \alglinref{line:measure edge weight} and \alglinref{line:KM} in \algref{alg:OMUA}, respectively. The SNR-based scheme has the complexity of $O\left(\left| U \right| \left( \left|\mathcal{C}_\text{m}\right| \left|\mathcal{B}_\text{m} \right| + \left|\mathcal{C}_\text{p}\right| \left|\mathcal{B}_\text{p} \right|\right)\right)$, whereas the SINR-based scheme needs to traverse all entities to identify potential interferers for each connection, so it has the complexity of $O\left(\left| U \right| \left( \left|\mathcal{C}_\text{m}\right| \left|\mathcal{B}_\text{m} \right| + \left|\mathcal{C}_\text{p}\right| \left|\mathcal{B}_\text{p} \right| \right) \left( \left| U \right| + \left|\mathcal{B}\right| \right) \right)$. The Kuhn-Munkres algorithm can find an optimal matching with the complexity of $O\left( \max \left\{\left| U \right|, \left|\mathcal{C}_\text{m}\right| \left|\mathcal{B}_\text{m} \right| + \left|\mathcal{C}_\text{p}\right| \left|\mathcal{B}_\text{p} \right|\right\}^3 \right)$. Thus, the time complexity of OMUA is $O\left( \max \left\{\left| U \right|, \left|\mathcal{C}_\text{m}\right| \left|\mathcal{B}_\text{m} \right| + \left|\mathcal{C}_\text{p}\right| \left|\mathcal{B}_\text{p} \right|\right\}^3 \right)$.

Determining the switching point for all BSs requires traversing all MUEs and BSs once, so the complexity is $O\left( \left| \mathcal{U} \right| + \left| \mathcal{B} \right| \right)$.

The analysis of SCSA is relatively cumbersome. Measuring the interference power between all MUE connections has the complexity of $O\left( \left| \mathcal{U} \right|^2 \right)$. Collapsing the interference graph also requires checking if any two MUEs need to be collapsed, which has the complexity of $O\left( \left| \mathcal{U} \right|^2 \right)$. Mapping the collapsed interference matrix to a similarity matrix requires processing at most $\left| \mathcal{U} \right|^2$ elements, it also has the complexity of $O\left( \left| \mathcal{U} \right|^2 \right)$. Computing the Laplacian has the complexity of $O\left( \left| \mathcal{U} \right|^2 \right)$. Finding the eigenvalues of the Laplacian has the complexity of $O\left( \left| \mathcal{U} \right|^3 \right)$, and sorting them in ascending order of corresponding eigenvalues has the complexity of $O\left( \left| \mathcal{U} \right| \log_2 \left| \mathcal{U} \right| \right)$. K-means has the complexity of $O\left( \left| \mathcal{U} \right| \left| \mathcal{C} \right| \right)$ per iteration and only a few iterations to converge in practice. Searching all clusters for conflicting MUEs has the complexity of $O\left( \left| \mathcal{U} \right|^2 \right)$, finding target clusters and increments of cut-edge weights for all such MUEs has the complexity of $O\left( \left| \mathcal{U} \right|^2 \right)$, and sorting them in descending order of increment has the complexity of $O\left( \left| \mathcal{U} \right| \log_2 \left| \mathcal{U} \right| \right)$. Assigning subchannel for each cluster after correction has the complexity of $O\left( \left| \mathcal{U} \right| + \left| \mathcal{C} \right|^2 \right)$. Considering that the number of MUEs in the system is often much higher than that of subchannels at the BS, the time complexity of SCSA is bounded by $O\left( \left| \mathcal{U} \right|^3 \right)$.

Therefore, the time complexity of OMSC consisting of the serial execution of these subalgorithms is $O\left( \max \left\{\left| U \right|, \left|\mathcal{C}_\text{m}\right| \left|\mathcal{B}_\text{m} \right| + \left|\mathcal{C}_\text{p}\right| \left|\mathcal{B}_\text{p} \right|\right\}^3 \right)$.

\section{Numerical Simulations} \label{sec:numerical simulations}

In this section, we begin to simulate the performance of the algorithm proposed in this paper and previous work in solving (P1). First, we introduce the comparison methods adopted, and then describe the performance metrics on which the simulation focuses. Afterwards, we specify the simulation environment and parameter settings and finally discuss and analyze the results obtained.

\begin{table}[!tb]
	\begin{center}
		\caption{Parameter settings}
		\label{tab:parameter settings}
		\begin{tabular}{ p{5em} p{19em} }
			\hline
			\textbf{Symbol} & \textbf{Value} \\
			\hline
			$N_\text{s}$, $\left|\mathcal{U}\right|$ & 8, \{150, 200, 250\} \\
			$\left|\mathcal{B}_\text{m}\right|, \left|\mathcal{B}_\text{p}\right|$ & 2, 60 \\
			$\left|\mathcal{C}_\text{m}\right|, \left|\mathcal{C}_\text{p}\right|$ & 18 \cite{alizadeh2022reinforcement}, 3  \\
			$ f_{\text{LTE}}, f_{\text{mmW}}$ & 1.9 GHz, 28 GHz \cite{qiao2023joint} \\
			$W_\text{m}, W_\text{p}$ & 1.8 MHz, 14.4 MHz \cite{qiao2023joint} \\
			$G_b, \theta_b$ & \{ (24.5 dBi, 10$^\circ$) \cite{rangan2014millimeter}, (15 dBi, 30$^\circ$) \cite{rappaport2013millimeter} \} \\
			$\vartheta_b, \vartheta_u$ & 90$^\circ$, 90$^\circ$ \cite{shokri2015beam} \\
			$n_\text{L}^{\text{LTE}}, n_\text{N}^{\text{LTE}}$ & 2, 3.37 \cite{blackard1992path, nwelih2022method} \\
			$n_\text{L}^{\text{mmW}}, n_\text{N}^{\text{mmW}}$ & 2.55, 5.76 \cite{rappaport2013millimeter} \\
			$P_\text{m}, P_\text{p}, P_u$ & 43 dBm, 33 dBm \cite{lee2022message}, 30 dBm \cite{dai2022joint} \\
			$N_\text{0}$ & $-$174 dBm/Hz \cite{alizadeh2022reinforcement} \\
			$\lambda_\text{o}$, $\mathbb{E}_\text{L}$ & $4.4 \times 10^{-4} / \text{m}^2$ \cite{bai2014analysis}, 55 m \cite{bai2014analysis} \\
			$T_{\text{p}}, T_\text{s}$ & 20 \textmu s, 65535 \textmu s \cite{shokri2015beam} \\
			$\beta_\text{f}, \beta_\text{r}$ & \{0.5, 1.0, 1.5\}, 0.5 \cite{rhee2011levy} \\
			$R^{\text{UL}}_u, R^{\text{DL}}_u$ & \{(15 Mbps, 1.0 Mbps), (15 Mbps, 15 Mbps), \hspace*{2pt} (0.1 Mbps, 15 Mbps)\} \\
			\hline
		\end{tabular}
	\end{center}
\end{table}

\subsection{Comparison Algorithms}

We implement two schemes of the OMSC algorithm, the default adopts the SNR-based pseudo supply-demand ratio, and we denote the one based on SINR as OMSC SINR. In addition, we consider two baselines:
\begin{itemize}
	\item \textbf{LCUAS} \cite{zhang2020optimal}. It prioritizes the MUEs with the least number of available BSs around them to establish connections, and the association rule is to select the BS with the least number of associated MUEs among the available BSs, i.e., the BS with the lightest load.
	
	\item \textbf{SDMAB} \cite{alizadeh2022reinforcement}. It allows each MUE to maintain an MAB, where arms correspond to individual BSs. The MUE randomly selects the BSs to be associated using the upper confidence bound principle. The central controller collects all the MUEs' association requests and evaluates the system performance under both the new association scheme and the historical optimal association scheme. The controller selects the better one as the actual scheme to be executed and delivers it to MUEs. MUEs update the parameters of their MABs based on the subsequent communication performance.
\end{itemize}

It is worth noting that neither LCUAS nor SDMAB consider uplink and downlink switching and subchannel allocation. Hence, we complement them with a baseline scheme, where the switching point is fixed to the midpoint of the time slot and the first idle subchannel at the BS is allocated to the newly associated MUE, to make them applicable to our problem.

\begin{figure*}[!t]
	\centering
	\subfloat[$\left|\mathcal{U}\right| = 150, \theta_b = 30^\circ$]{
		\includegraphics[width=0.33\linewidth]{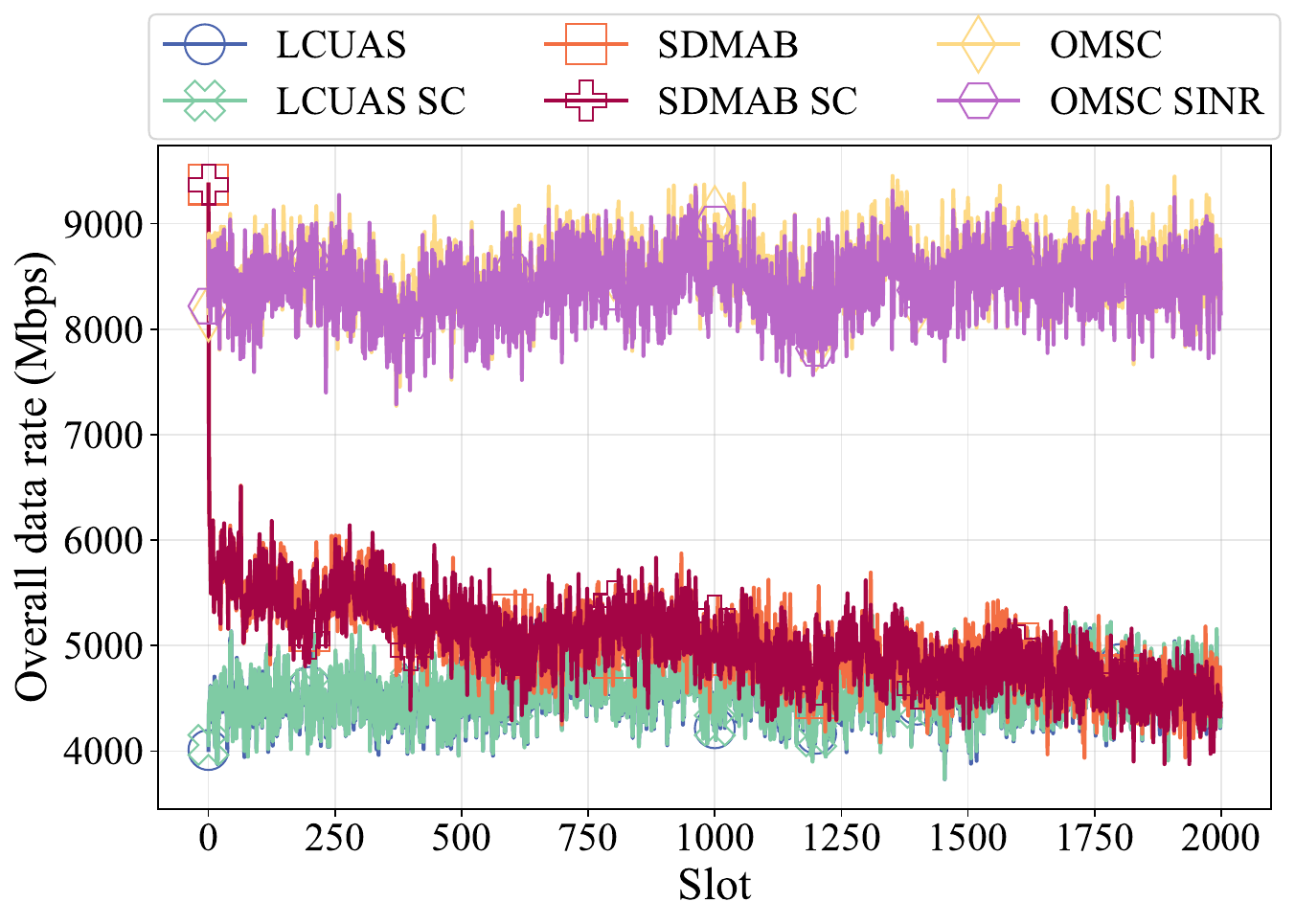}
		\label{fig:overall data rate U 150 theta 30 beta 0.5}
	}
	\subfloat[$\left|\mathcal{U}\right| = 200, \theta_b = 30^\circ$]{
		\includegraphics[width=0.33\linewidth]{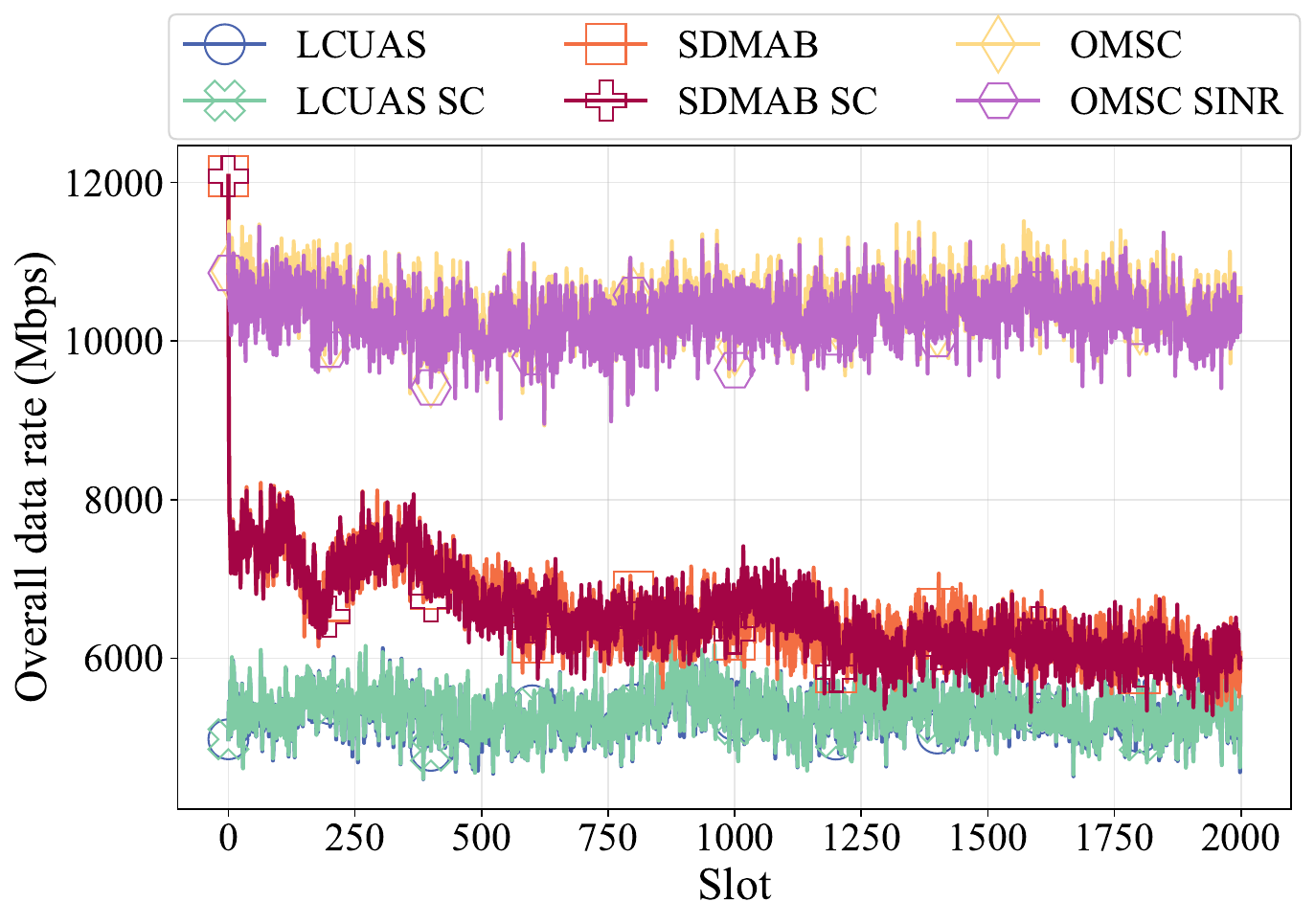}
		\label{fig:overall data rate U 200 theta 30 beta 0.5}
	}
	\subfloat[$\left|\mathcal{U}\right| = 250, \theta_b = 30^\circ$]{
		\includegraphics[width=0.33\linewidth]{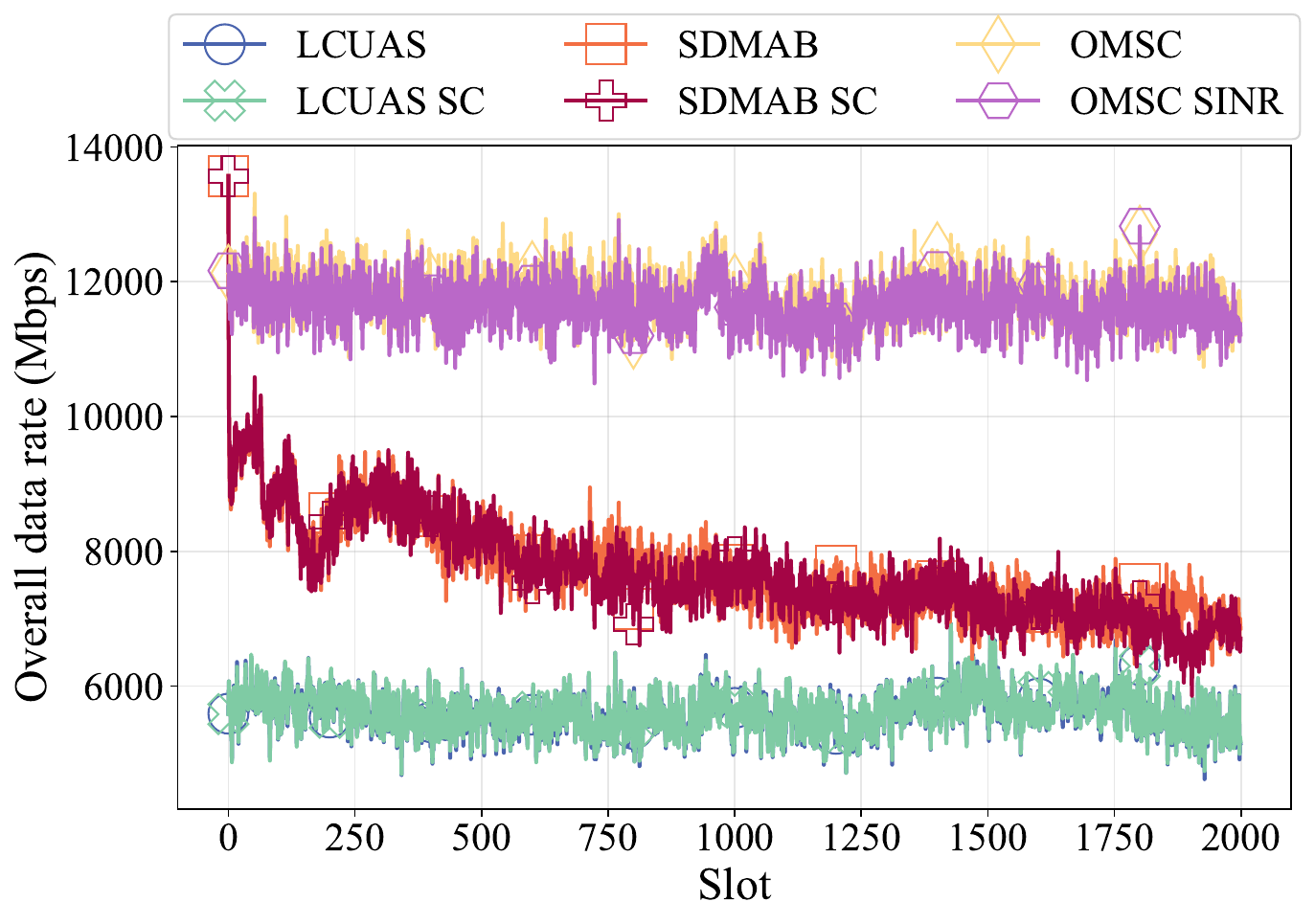}
		\label{fig:overall data rate U 250 theta 30 beta 0.5}
	}
	\\
	\subfloat[$\left|\mathcal{U}\right| = 150, \theta_b = 10^\circ$]{
		\includegraphics[width=0.33\linewidth]{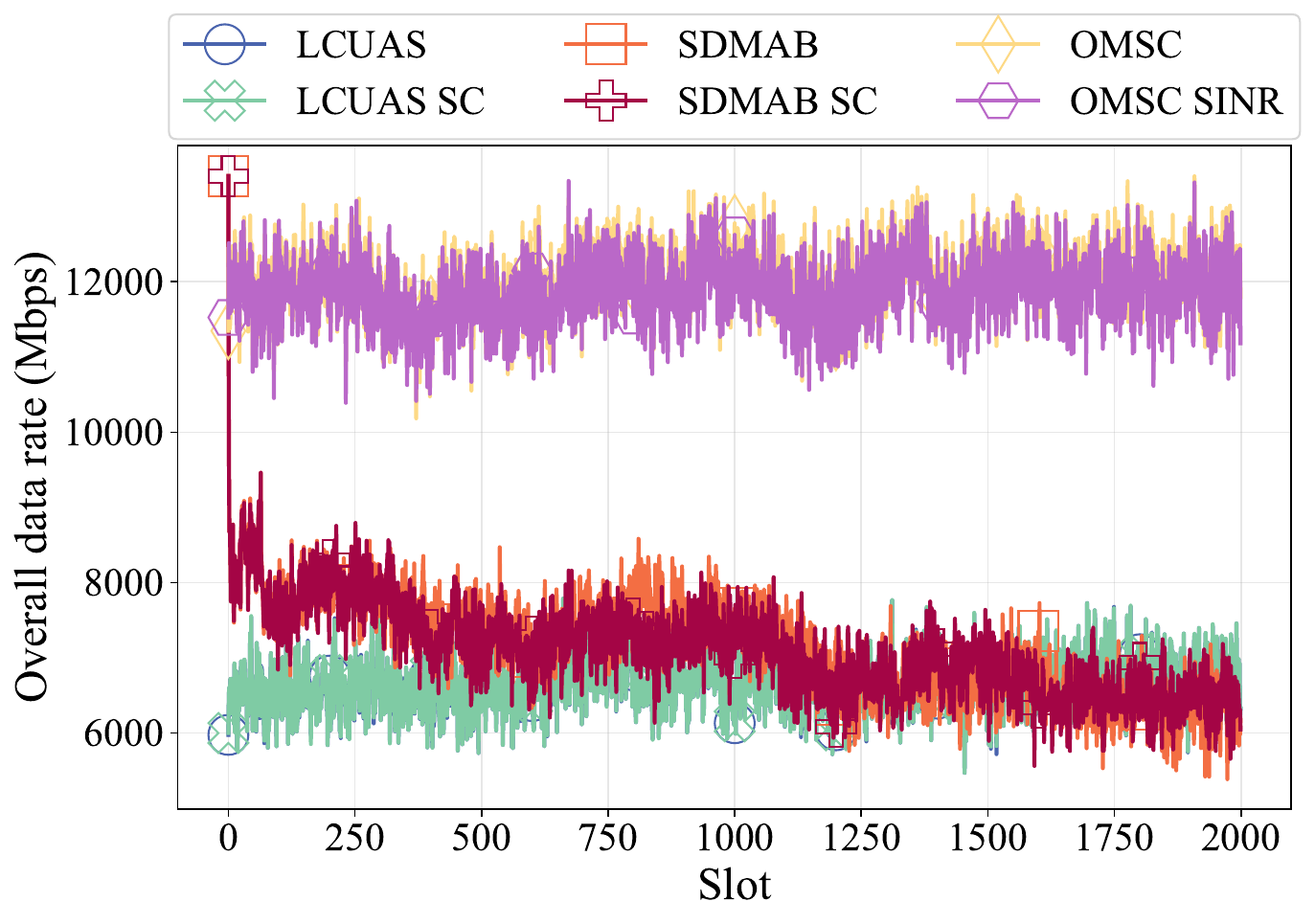}
		\label{fig:overall data rate U 150 theta 10 beta 0.5}
	}
	\subfloat[$\left|\mathcal{U}\right| = 200, \theta_b = 10^\circ$]{
		\includegraphics[width=0.33\linewidth]{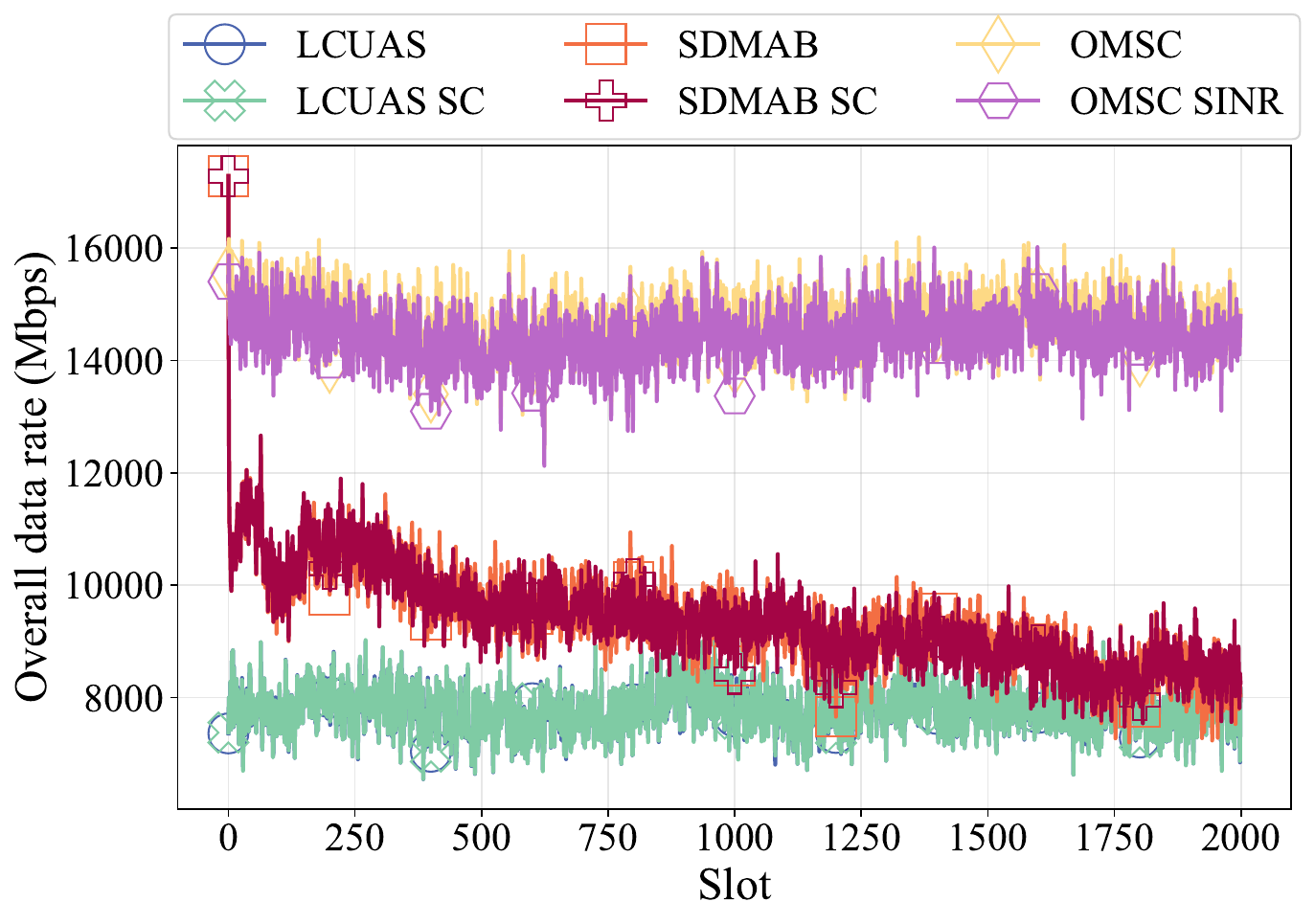}
		\label{fig:overall data rate U 200 theta 10 beta 0.5}
	}
	\subfloat[$\left|\mathcal{U}\right| = 250, \theta_b = 10^\circ$]{
		\includegraphics[width=0.33\linewidth]{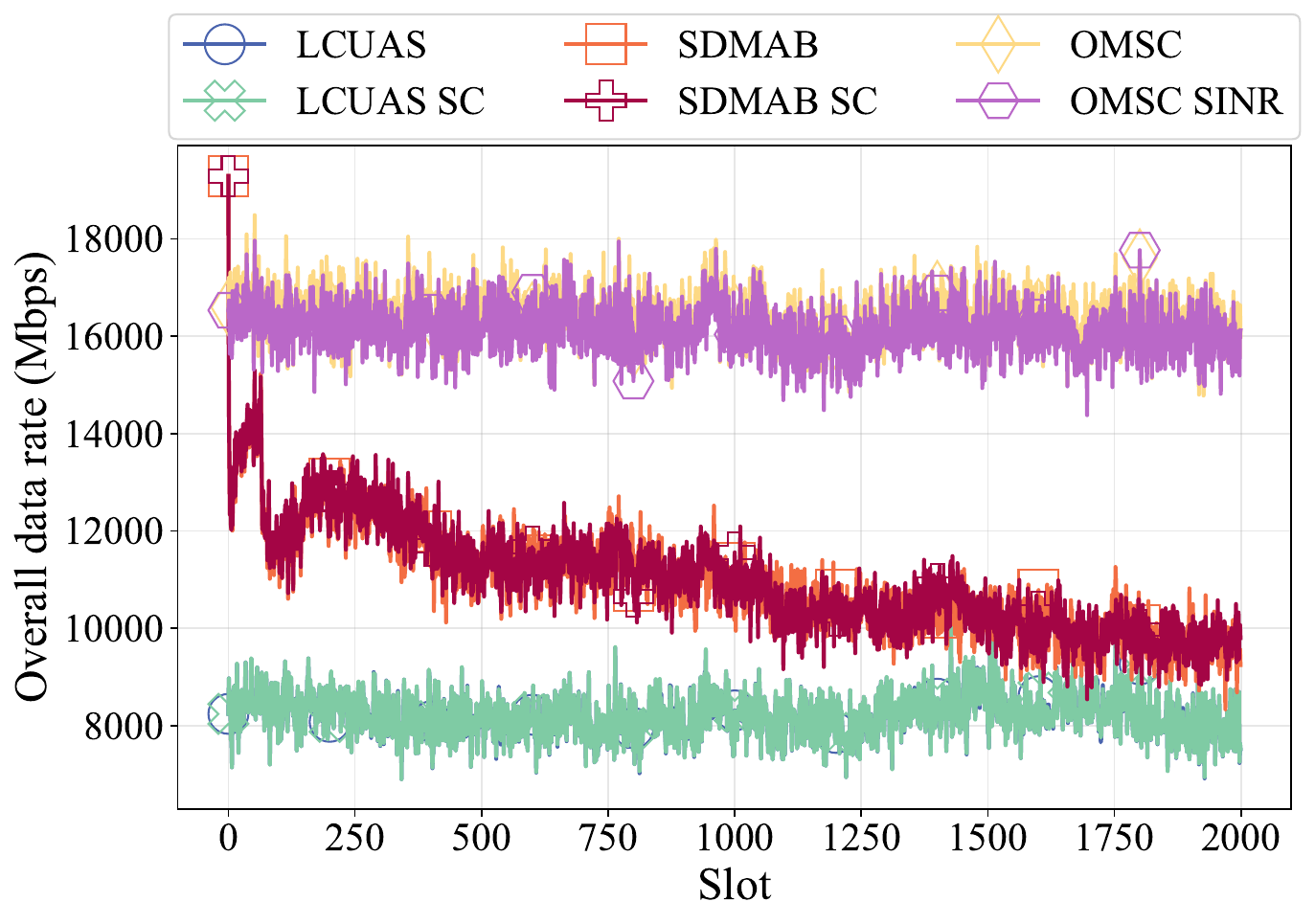}
		\label{fig:overall data rate U 250 theta 10 beta 0.5}
	}
	\caption{Overall data rate when $\beta_\text{f}$ is fixed at 0.5.}
	\label{fig:overall data rate}
\end{figure*}

\subsection{Performance Metrics}

We focus on the following performance indicators:
\begin{enumerate}
	\item \textbf{Overall data rate.} It refers to the sum of the uplink and downlink data rates for all MUEs in the system.
	
	\item \textbf{Effective data rate.} It refers to the sum of the data rates of MUEs whose requirements are met in both directions.
	
	\item \textbf{The number of satisfied MUEs.} It refers to the number of MUEs whose requirements are met in both directions.
	
	\item \textbf{Decision time.} It refers to the average runtime consumed by the algorithm to make a single decision.
\end{enumerate}

\subsection{Simulation Environment and Parameter Settings}

The simulation is deployed on a server equipped with 1051 Intel(R) Xeon(R) Gold 5218R CPU @ 2.10GHz and 252G memory. The programming language is Python 3.10.4.

The simulation area is a rectangle with horizontal axis ranges $[-1000, 1000]$ meters and vertical axis ranges $[-500, 500]$ meters. There are two MBSs located at $(-500, 0)$ and $(500, 0)$, and PBSs scatter uniformly in the area at a density of $30/\text{km}^2$. We summarize the values of other main parameters in \tabref{tab:parameter settings}. Symmetrical mmWave antennas are considered for BSs and MUEs \cite{rappaport2013millimeter}. The ratio of MUEs for each data rate requirement is set to $3:4:3$. For each parameter combination, we performed simulations up to 2000 time slots with 10 repetitions.

\subsection{Performance Evaluation}

In this subsection, we show the performance of each algorithm. Among them, LCUAS SC replaces the baseline subchannel assignment of LCUAS for SCSA, and so does SDMAB SC. 

\subsubsection{Overall data rate}

\begin{figure*}[!t]
	\centering
	\subfloat[$\left|\mathcal{U}\right| = 150, \theta_b = 30^\circ$]{
		\includegraphics[width=0.33\linewidth]{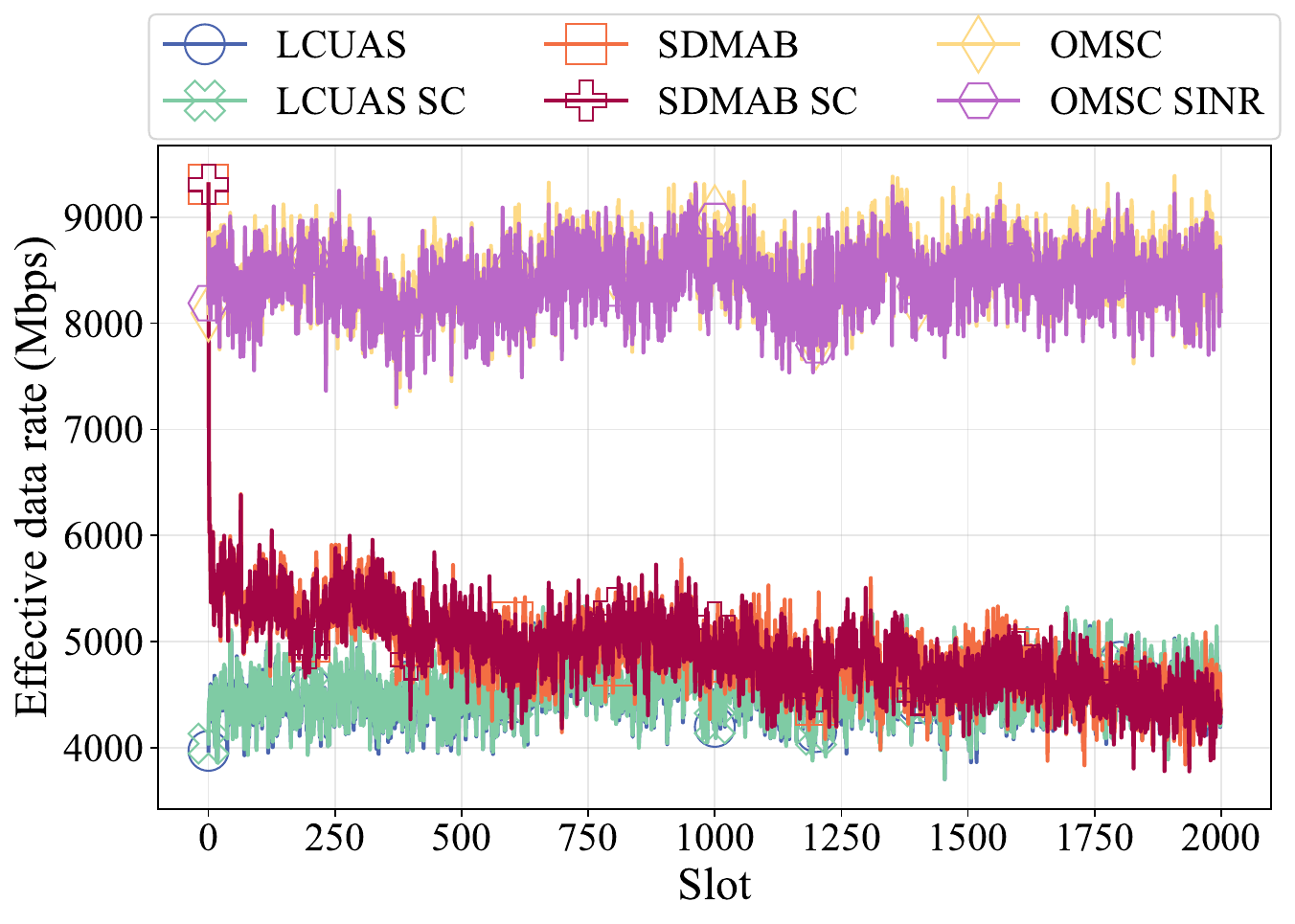}
		\label{fig:effective data rate U 150 theta 30 beta 0.5}
	}
	\subfloat[$\left|\mathcal{U}\right| = 200, \theta_b = 30^\circ$]{
		\includegraphics[width=0.33\linewidth]{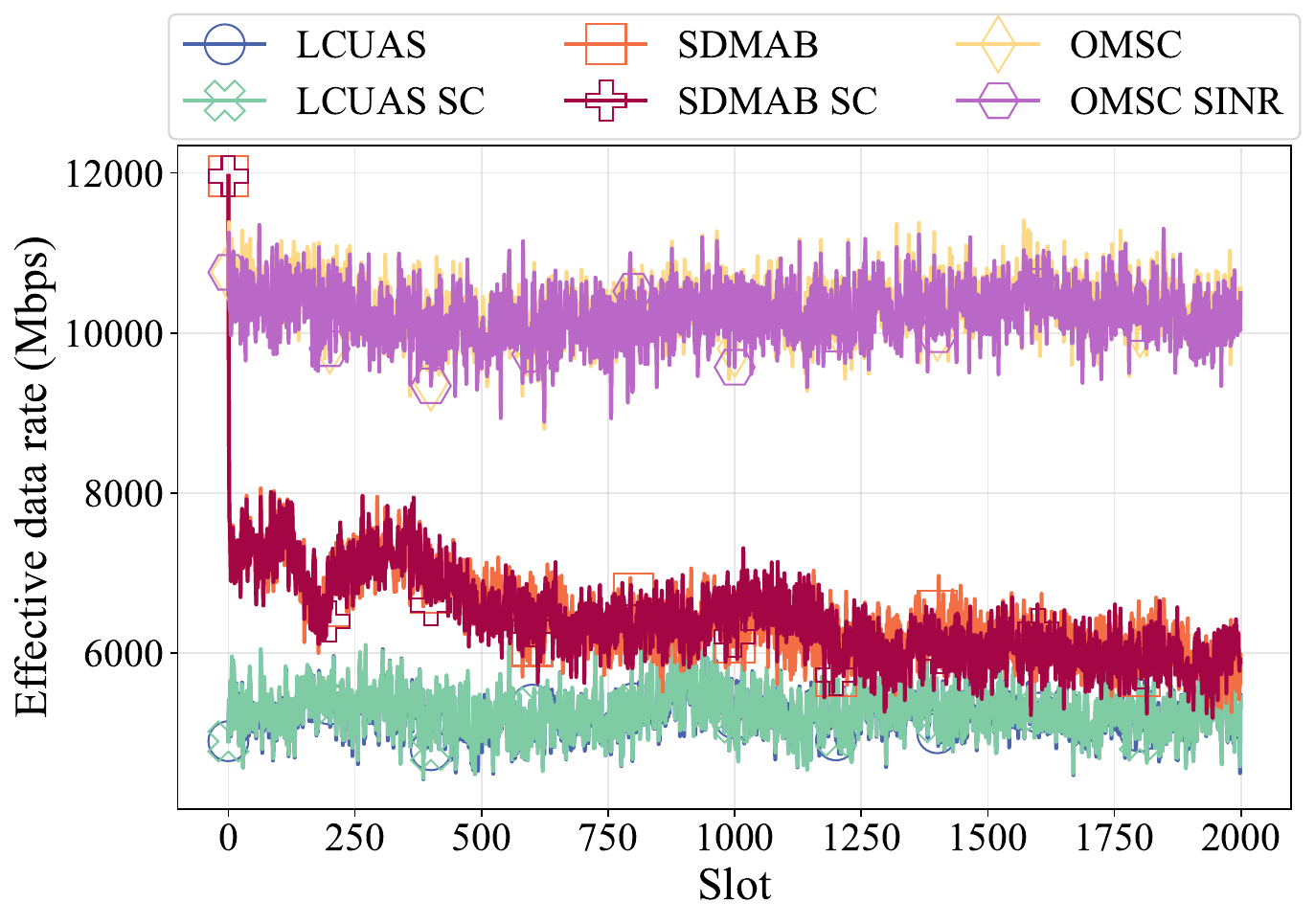}
		\label{fig:effective data rate U 200 theta 30 beta 0.5}
	}
	\subfloat[$\left|\mathcal{U}\right| = 250, \theta_b = 30^\circ$]{
		\includegraphics[width=0.33\linewidth]{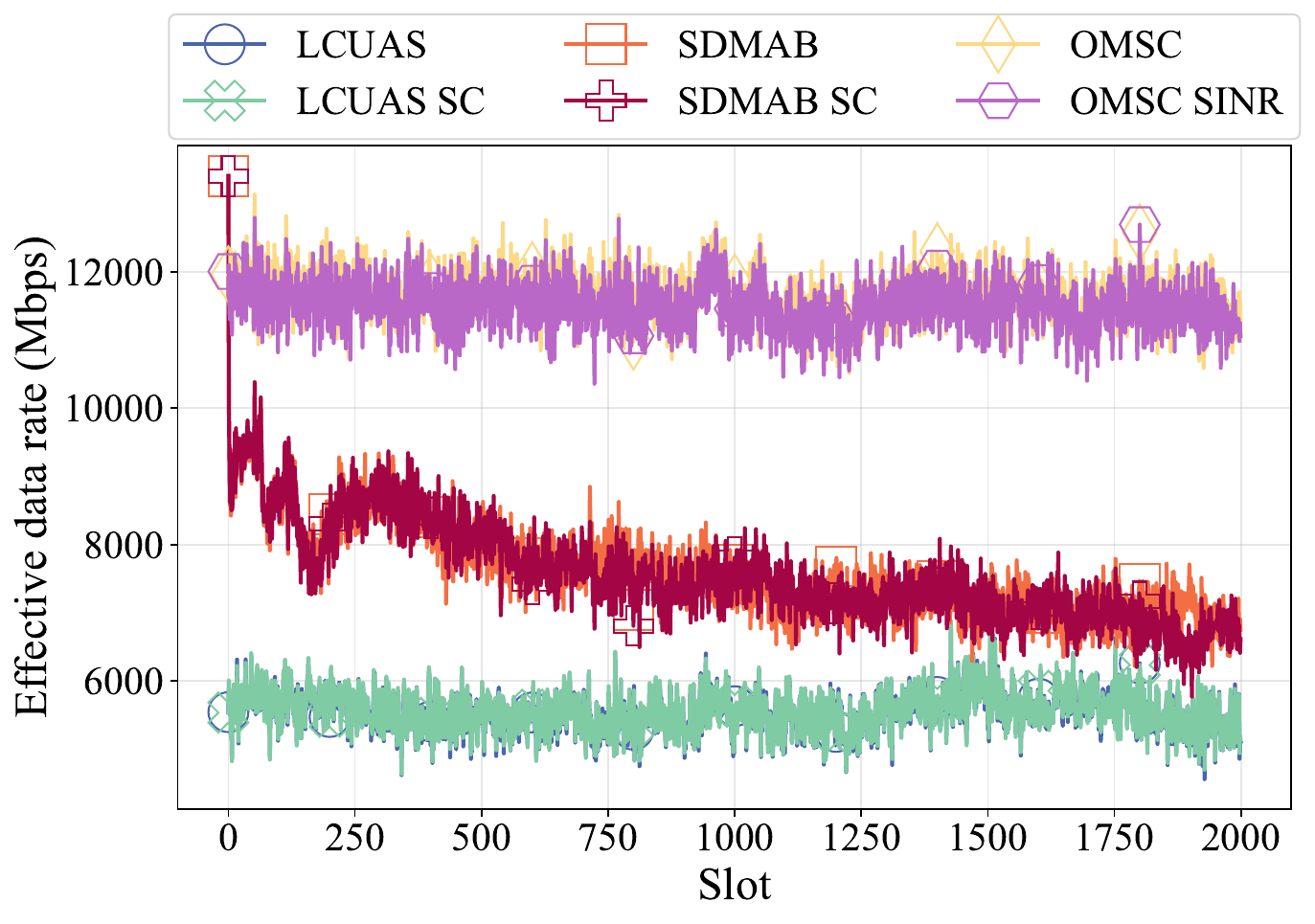}
		\label{fig:effective data rate U 250 theta 30 beta 0.5}
	}
	\\
	\subfloat[$\left|\mathcal{U}\right| = 150, \theta_b = 10^\circ$]{
		\includegraphics[width=0.33\linewidth]{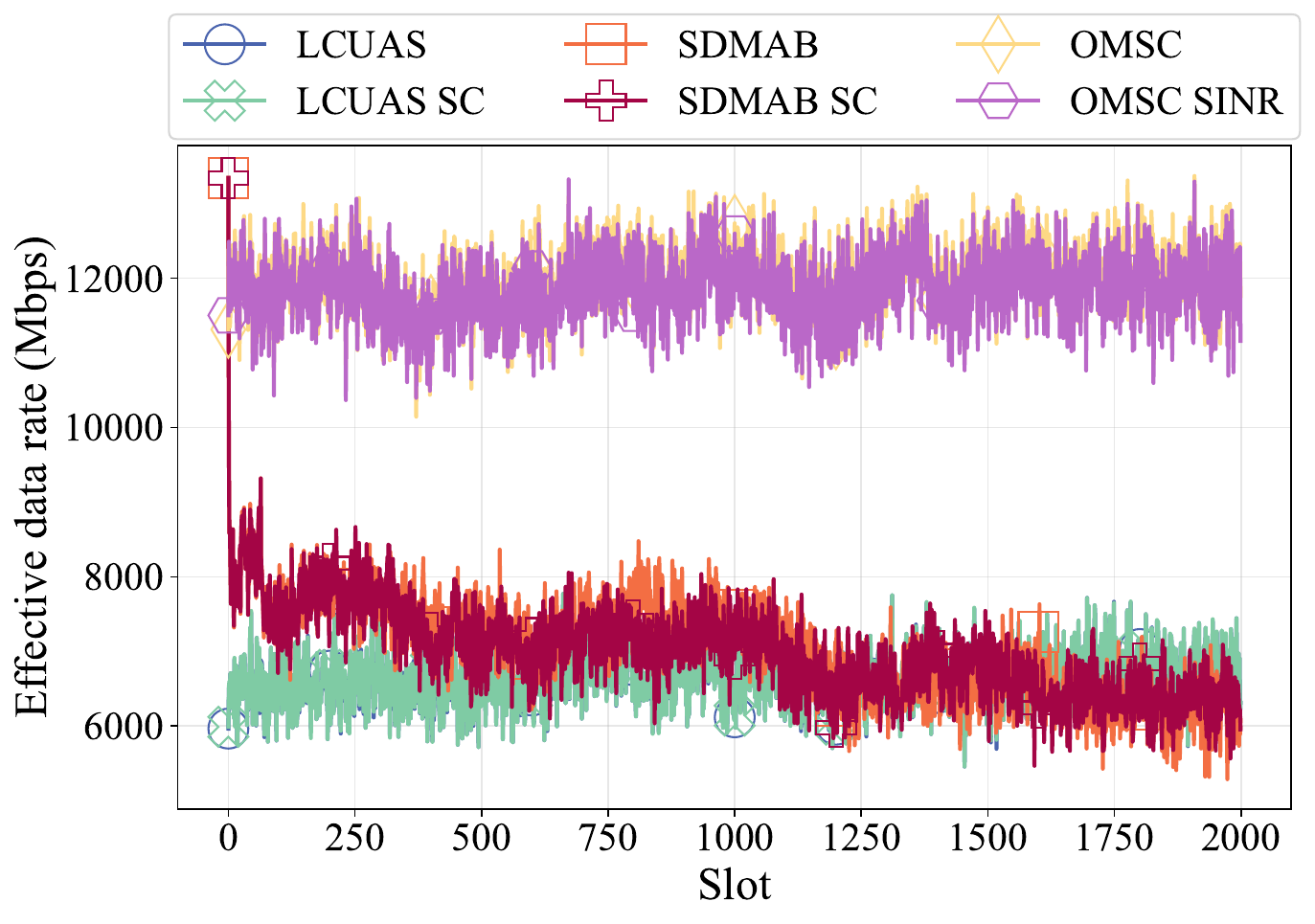}
		\label{fig:effective data rate U 150 theta 10 beta 0.5}
	}
	\subfloat[$\left|\mathcal{U}\right| = 200, \theta_b = 10^\circ$]{
		\includegraphics[width=0.33\linewidth]{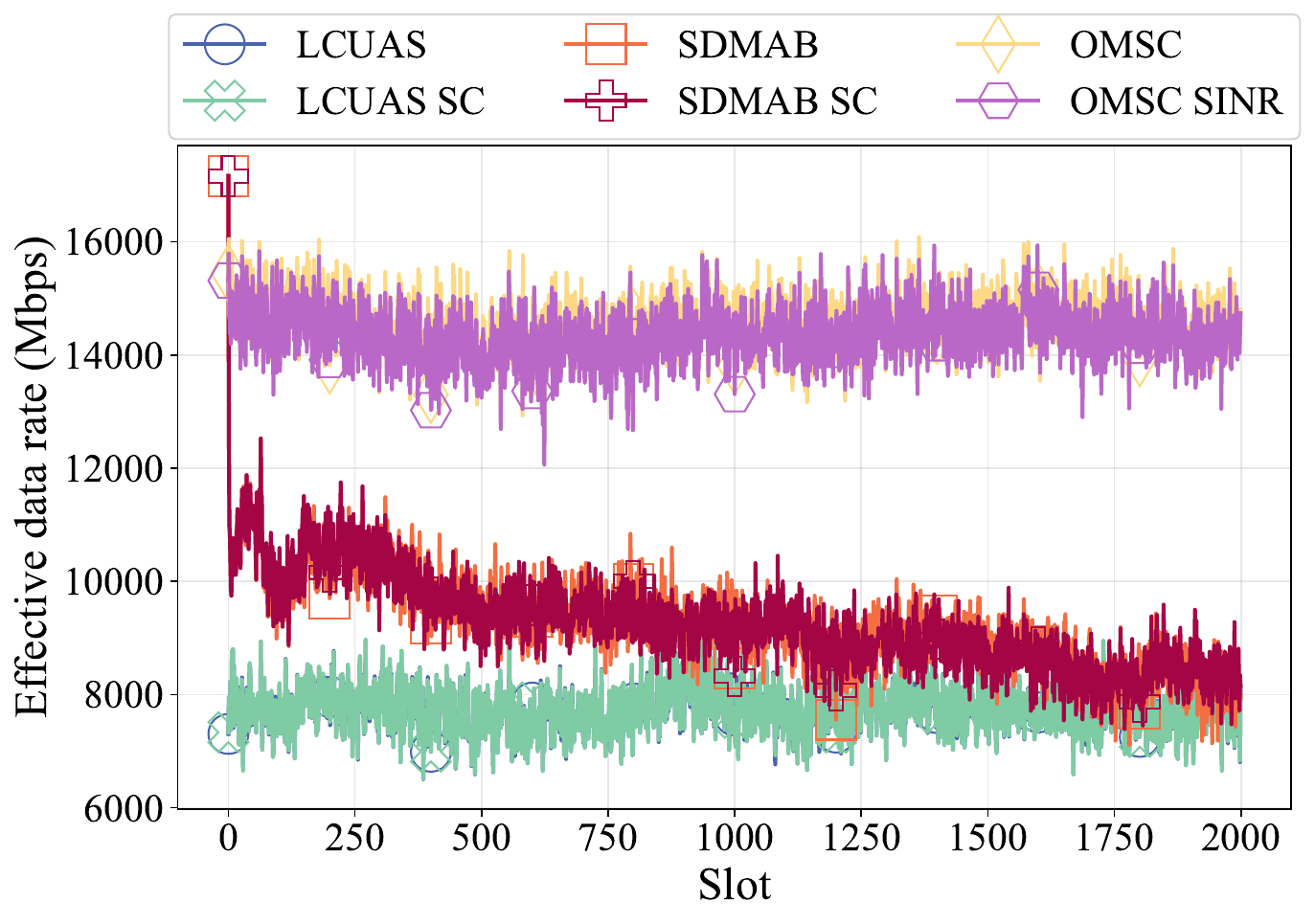}
		\label{fig:effective data rate U 200 theta 10 beta 0.5}
	}
	\subfloat[$\left|\mathcal{U}\right| = 250, \theta_b = 10^\circ$]{
		\includegraphics[width=0.33\linewidth]{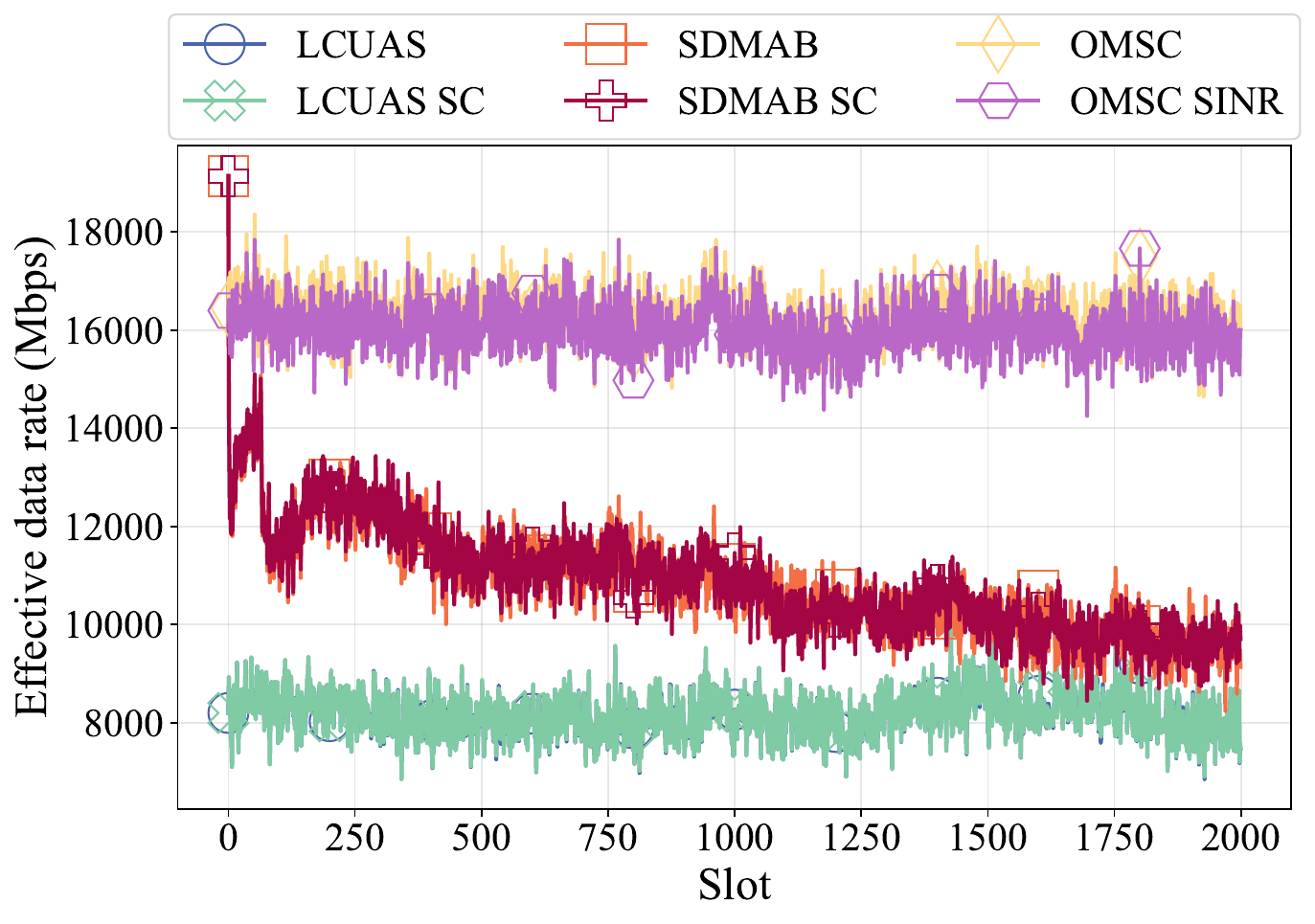}
		\label{fig:effective data rate U 250 theta 10 beta 0.5}
	}
	\caption{Effective data rate when $\beta_\text{f}$ is fixed at 0.5.}
	\label{fig:effective data rate}
\end{figure*}

\figref{fig:overall data rate} presents the average of the overall data rate after repeated experiments. Different numbers of MUEs and antenna parameters are examined, but $\beta_\text{f}$ is fixed at 0.5. It is easy to see that the OMSC-related algorithms consistently outperform the others, and the LCUAS-related algorithms are the worst most of the time. This is mainly because LCUAS employs a coarse-grained UA scheme based on distance and BS load, which has the advantage of obtaining a feasible solution very quickly, but it is also prone to produce low-quality association schemes in multi-band HetNets. In contrast, SDMAB-related algorithms can interact with the environment based on the feedback from the central controller, which helps to improve the performance. However, the MAB requires the environment to be Markov, which is hardly guaranteed in multi-user mobile scenarios because the choices of all users collectively affect the environment, resulting in limited improvement.

\figref{fig:overall data rate} also demonstrates that in mmWave networks, the interference structure does correlate with the UA results. This is mainly reflected in the fact that the curves of LCUAS and LCUAS SC are highly overlapping, but those of SDMAB and SDMAB SC are significantly different. Comparing the curves of SDMAB and SDMAB SC, we can find that the lower envelope of the latter is generally higher than that of the former, and the range of curve fluctuations is more concentrated. This confirms our previous analysis. That is, the transformation and spectral clustering methods for the subchannel allocation problem do sacrifice the optimality of the solution, but SCSA can quickly give a not-bad subchannel assignment scheme and its performance is more stable than the benchmark.

Looking at \figref{fig:overall data rate} vertically, when the main beam of the directional antenna is narrowed and the directivity is increased, the overall data rate of all the algorithms is significantly improved, which confirms that the directional antenna can enhance the data transmission capability of mmWave networks. In addition, when the main beam is narrowed, the overlap of the curves of SDMAB and SDMAB SC becomes higher, which indicates that a directional antenna can indeed reduce the interference of mmWave links. This also explains why the OMSC SINR curve is almost always lower than that of OMSC, i.e., with the assistance of these hardware techniques and scheduling algorithms, the pseudo data rate based on SINR conservatively estimates the data transmission capability of the multi-band HetNet.

\subsubsection{Effective data rate}

In \figref{fig:effective data rate} we use the same parameter settings as in \figref{fig:overall data rate} to compare the effective data rates of the algorithms. The phenomenon shown in \figref{fig:effective data rate} is basically consistent with \figref{fig:overall data rate}, except that the effective data rate is lower than the corresponding overall data rate, which is almost imperceptible. From this phenomenon, it can be inferred that those MUEs who did not receive satisfactory service all experienced poor channel quality, so the contribution of their data rate to the overall data rate was almost negligible. This is very consistent with the communication characteristic of mmWave that poorer channels produce severe attenuation.

As shown in \figref{fig:overall data rate} and \figref{fig:effective data rate}, when the number of MUEs increases, the data rate achieved by each algorithm improves, which means that they are all able to filter the MUEs so that the limited access opportunities are granted to more suitable MUEs. This is exactly what the optimal matching does. Moreover, by feeding the pseudo supply-demand ratio into a submodular function to be the edge weights, OMSC effectively avoids the over-provisioning of user demand because the gains from increasing the supply are diminishing under the effect of the submodular function, allowing the matching results to satisfy more users and thus creating a significant performance advantage.

\subsubsection{Mumber of Satisfied MUEs}

\begin{figure*}[!tb]
	\centering
	\subfloat[$\left|\mathcal{U}\right| = 150, \beta_\text{f} = 0.5$]{
		\includegraphics[width=0.33\linewidth]{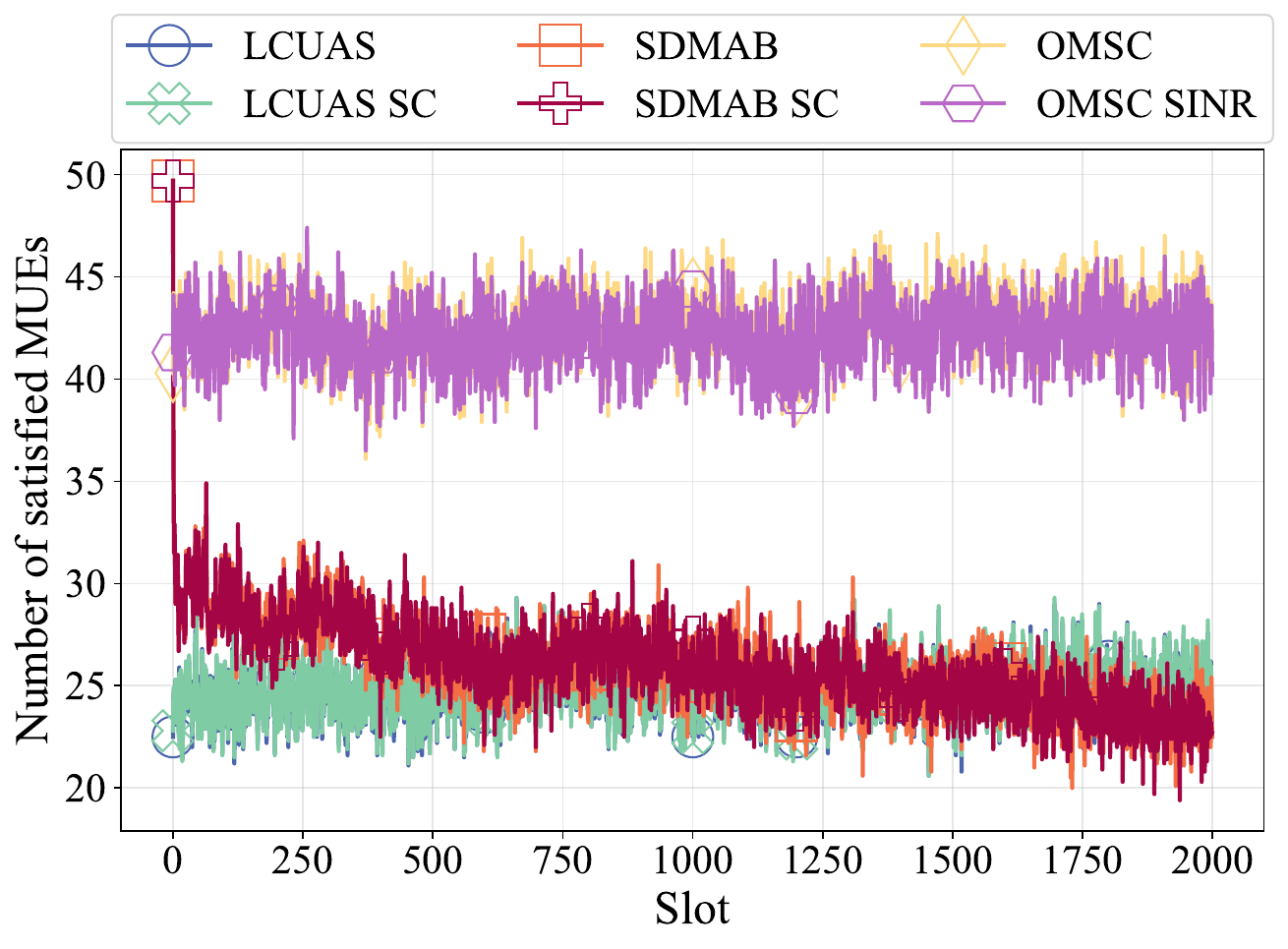}
		\label{fig:satisfied MUE number U 150 theta 30 beta 0.5}
	}
	\subfloat[$\left|\mathcal{U}\right| = 200, \beta_\text{f} = 0.5$]{
		\includegraphics[width=0.33\linewidth]{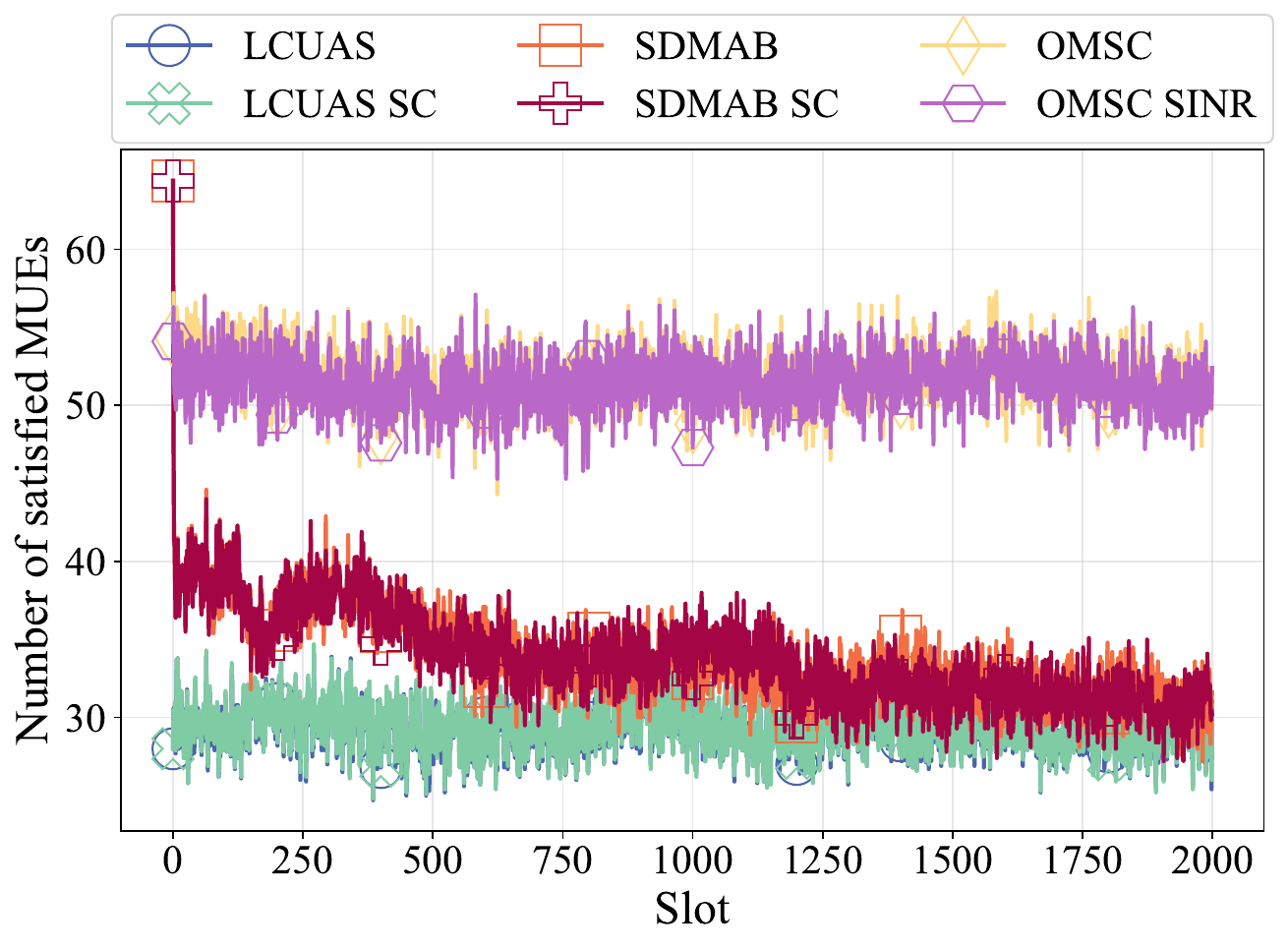}
		\label{fig:satisfied MUE number U 200 theta 30 beta 0.5}
	}
	\subfloat[$\left|\mathcal{U}\right| = 250, \beta_\text{f} = 0.5$]{
		\includegraphics[width=0.33\linewidth]{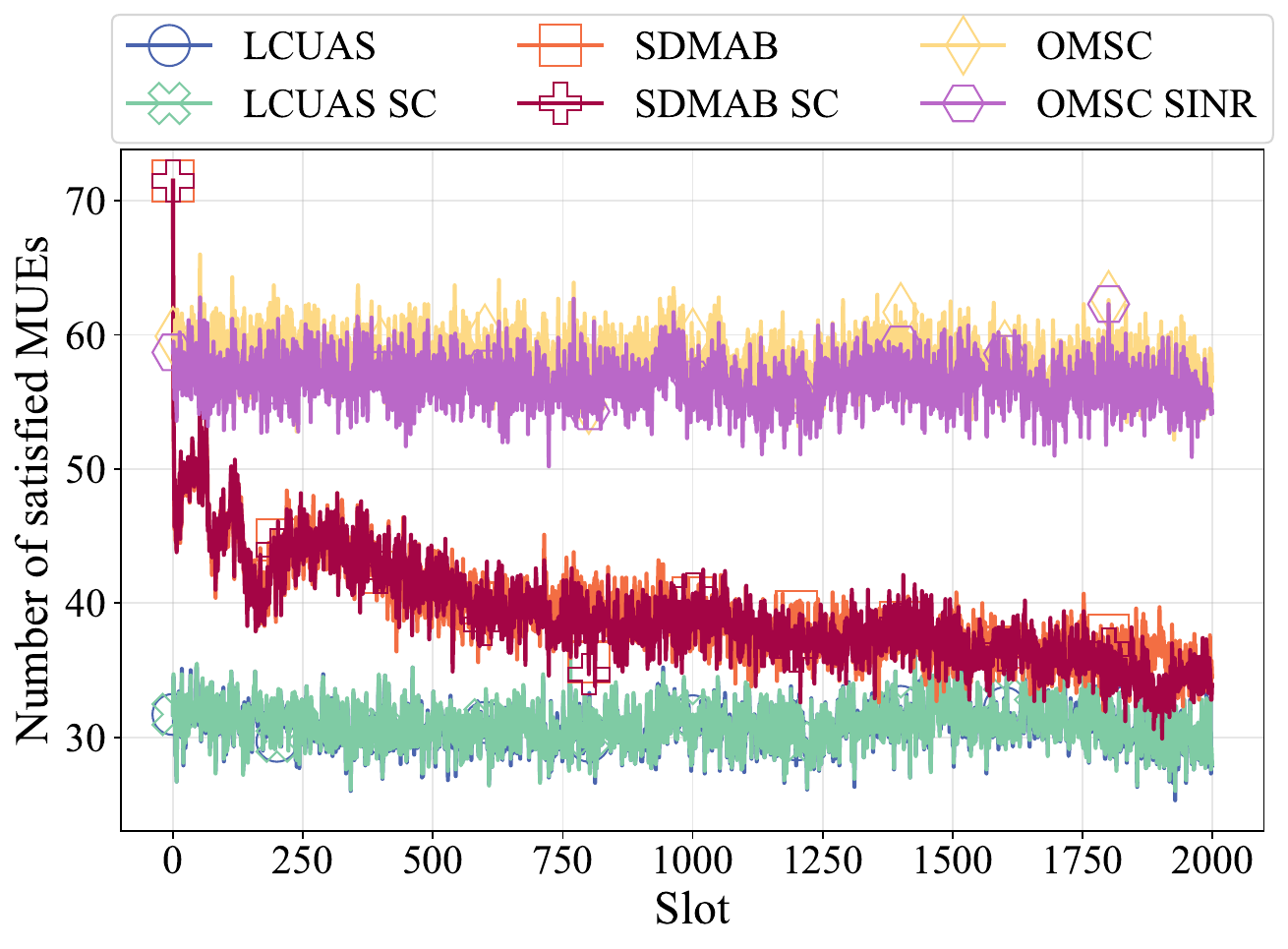}
		\label{fig:satisfied MUE number U 250 theta 30 beta 0.5}
	}
	\\
	\subfloat[$\left|\mathcal{U}\right| = 150, \beta_\text{f} = 1.5$]{
		\includegraphics[width=0.33\linewidth]{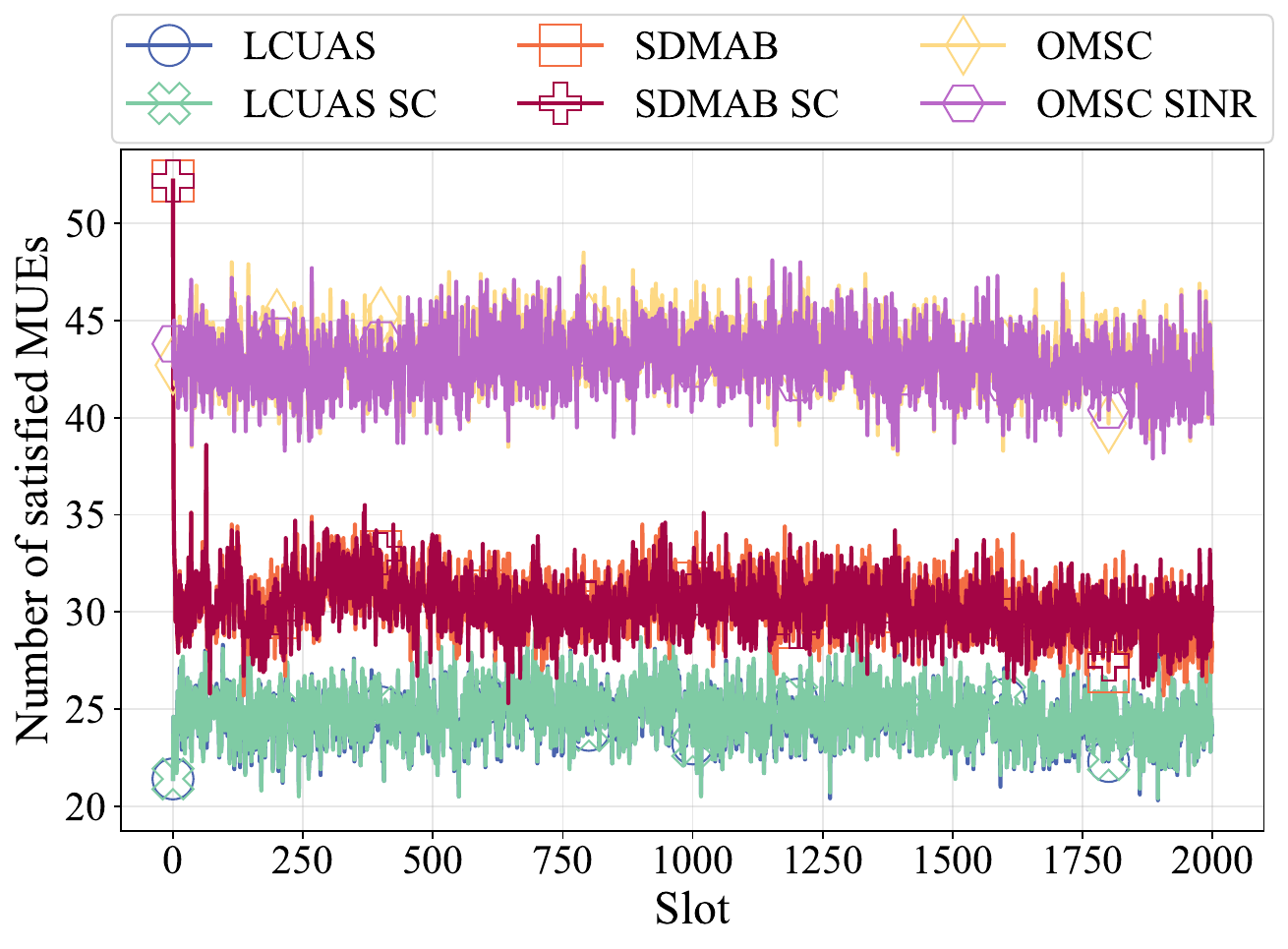}
		\label{fig:satisfied MUE number U 150 theta 30 beta 1.5}
	}
	\subfloat[$\left|\mathcal{U}\right| = 200, \beta_\text{f} = 1.5$]{
		\includegraphics[width=0.33\linewidth]{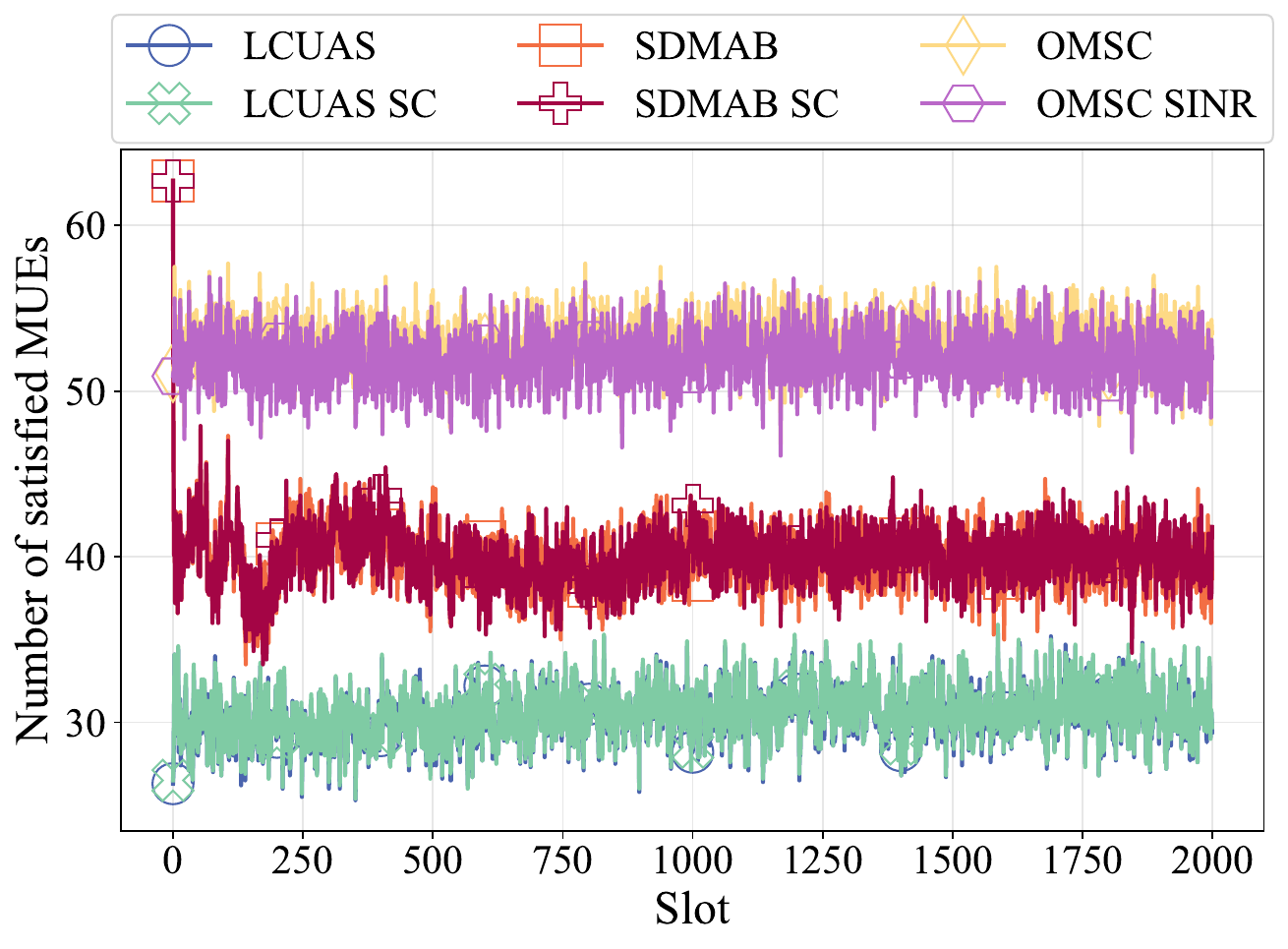}
		\label{fig:satisfied MUE number U 200 theta 30 beta 1.5}
	}
	\subfloat[$\left|\mathcal{U}\right| = 250, \beta_\text{f} = 1.5$]{
		\includegraphics[width=0.33\linewidth]{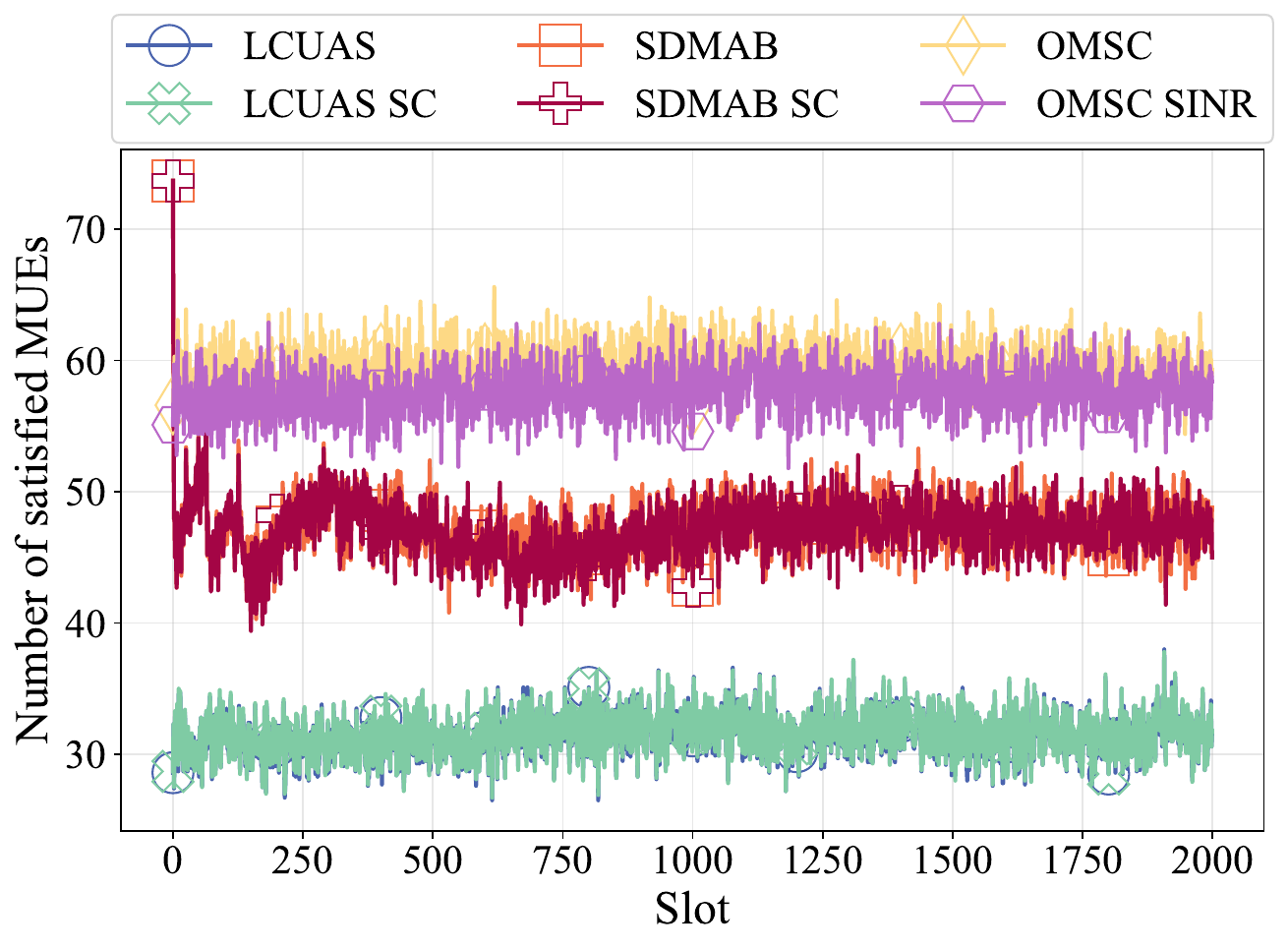}
		\label{fig:satisfied MUE number U 250 theta 30 beta 1.5}
	}
	\caption{Number of satisfied MUEs when $\theta_b$ is fixed at 30$^\circ$.}
	\label{fig:satisfied MUE number}
\end{figure*}

\figref{fig:satisfied MUE number} illustrates the number of satisfied MUEs for the different total numbers of MUEs and Levy flight exponent, the directional antenna parameters are fixed to $(15\text{ dBi}, 30^\circ)$. The relative relationships of the curves in \figref{fig:satisfied MUE number U 150 theta 30 beta 0.5} to \figref{fig:satisfied MUE number U 250 theta 30 beta 0.5} are similar to the first three graphs in \figref{fig:overall data rate} and \figref{fig:effective data rate}, which again supports the above analysis.

Comparing each column in \figref{fig:satisfied MUE number} shows that when user mobility becomes weaker, i.e., $\beta_\text{f}$ increases, the performance of LCUAS-related and OMSC-related algorithms are only slightly promoted, while that of the SDMAB-related algorithms are significantly enhanced. This is mainly because the Markov property of the environment becomes more stringent when user mobility decreases, which obviously has a greater impact on the SDMAB-related algorithms. However, scenarios, where dynamics exist and multiple users make decisions independently, can hardly be strictly Markov, and the inherent shortcomings of SDMAB always exist, so its performance still cannot break through the level of OMSC.

For the other two classes of algorithms, although they do not depend on the Markov property of the environment, the uncertainty of the wireless channel is relatively moderated by the reduced mobility of the users, so their performance is also slightly improved. From another perspective, it implies that the performance of these algorithms is more robust to user mobility.

\subsubsection{Decision time}

The average time taken by each algorithm to make a single decision is given in \figref{fig:average decision time}. Unsurprisingly, the LCUAS algorithm is the fastest to solve, which stems from its concise execution at the expense of a more ordinary performance in other metrics. OMSC is slightly more efficient than SDMAB, mainly because the controller of SDMAB needs to try and determine which of the temporary UA results and the historical optimum yields a higher data rate in the current network state, which is a more time-consuming process. The additional runtime observed in LCUAS SC and SDMAB SC, compared to LCUAS and SDMAB respectively, is attributed to the greater complexity of the SCSA method relative to the baseline subchannel allocation scheme. With the positive effects of SCSA on other metrics, the additional time overhead introduced is acceptable and well-balanced.

\begin{figure}[!h]
	\centering
	\includegraphics[width=0.8\linewidth]{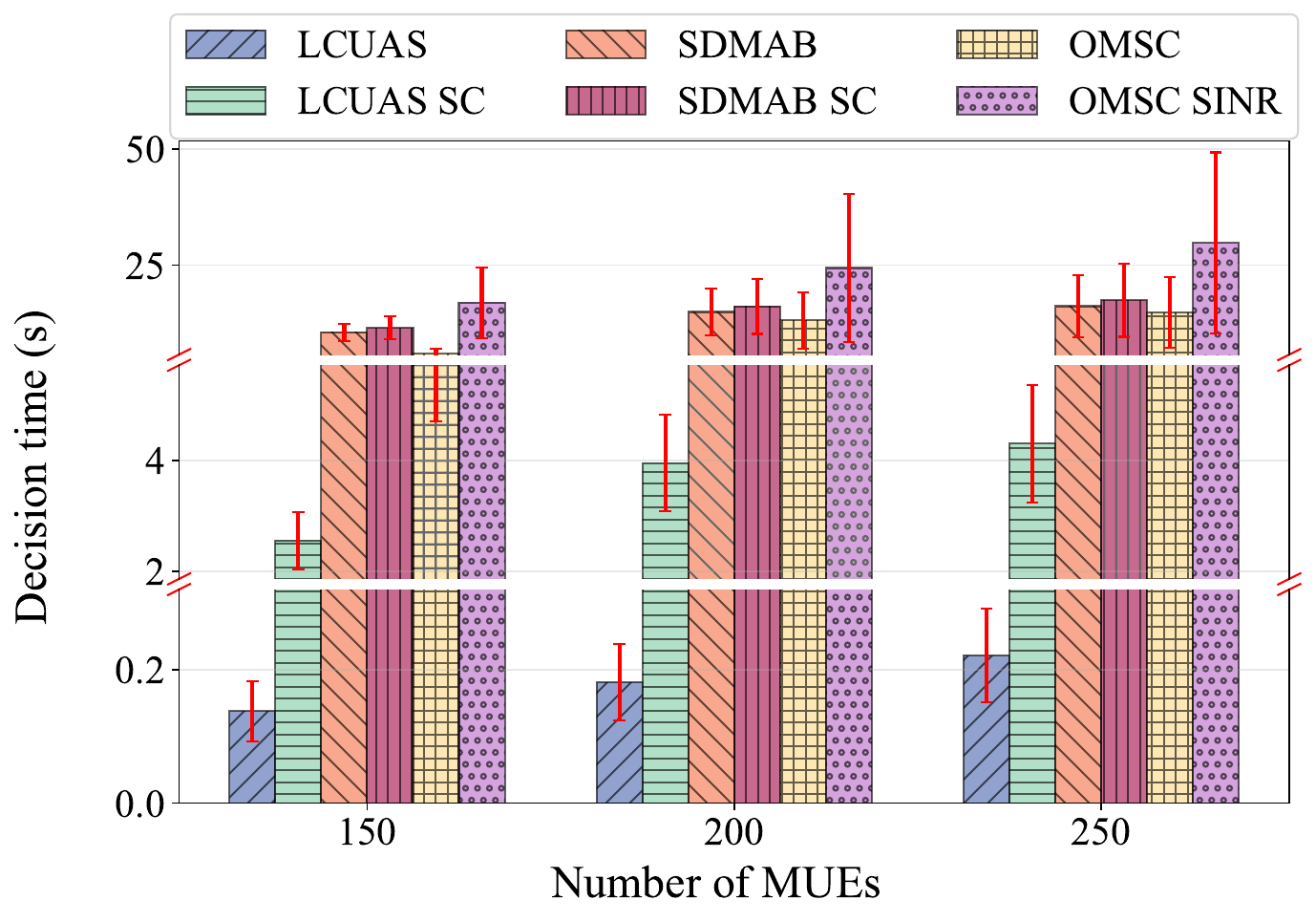}
	\caption{Average decision time.}
	\label{fig:average decision time}
\end{figure}

Referring to the impact of SCSA on the decision time of LCUAS and SDMAB, we can infer that more than half of the runtime of OMSC comes from OMUA, which is consistent with the analysis in \secref{sec:complexity analysis}. Moreover, the runtime of OMSC SINR is even higher than that of OMSC because it needs to additionally traverse all the entities for each connection to estimate the interference strength. However, as shown before, directional antennas and reasonable user association and subchannel scheduling are capable of suppressing the interference well, so pessimistically estimating the interference not only incurs an additional time cost, but also underestimates the transmission capability of the wireless HetNet, which in turn depreciates the algorithm performance. Therefore, we believe that OMSC is better suited than OMSC SINR to solve the problem of this work.

\section{Conclusion and Future Works} \label{sec:conclusion and future works}

In this paper, we investigated the wireless accessing problem in multi-band 5G HetNets, where the directional transmission model of mmWave, asynchronous uplink and downlink switching of BSs, mobility of service subscribers, and asymmetric service demands are considered. We formulated it as INLP and prove the intractability theoretically.

To solve the problem efficiently, we decoupled it into user association, switching point selection, and subchannel allocation problems. We also transformed them into the optimal matching of a bipartite and a generalized Max-cut problem for interference graph. We proposed the OMSC algorithm based on the Kuhn-Munkres algorithm and spectral clustering for the solution. Simulation results show that our algorithm outperforms the comparison methods in terms of overall data rate, effective data rate, and number of satisfied MUEs with acceptable runtime.

In addition, we confirmed that the mmWave interference structure is highly correlated with the UA results, and that directional antennas can effectively reduce mmWave interference. Moreover, in scenarios where both mmWave BSs and MUEs are densely distributed, mmWave interference does exist, but it can be suppressed or even avoided by using directional antennas, as well as improving UA and subchannel allocation.

From the simulation results, we can see that a fraction of the MUEs are not satisfied with the communication quality, which means that these connections are wasted. Since a single connection cannot satisfy the needs of these MUEs, a natural question is what would happen if multiple connections could be established simultaneously for a single MUE. Therefore, in the future, we intend to explore the impact of multi-connection techniques on multi-band 5G HetNets. It is foreseen that the scenario will become more complex but may also be a necessary step for the continuous evolution of the network.

\bibliographystyle{IEEEtran}
\bibliography{ref}

\begin{thebibliography}{10}
\providecommand{\url}[1]{#1}
\csname url@samestyle\endcsname
\providecommand{\newblock}{\relax}
\providecommand{\bibinfo}[2]{#2}
\providecommand{\BIBentrySTDinterwordspacing}{\spaceskip=0pt\relax}
\providecommand{\BIBentryALTinterwordstretchfactor}{4}
\providecommand{\BIBentryALTinterwordspacing}{\spaceskip=\fontdimen2\font plus
\BIBentryALTinterwordstretchfactor\fontdimen3\font minus
  \fontdimen4\font\relax}
\providecommand{\BIBforeignlanguage}[2]{{%
\expandafter\ifx\csname l@#1\endcsname\relax
\typeout{** WARNING: IEEEtran.bst: No hyphenation pattern has been}%
\typeout{** loaded for the language `#1'. Using the pattern for}%
\typeout{** the default language instead.}%
\else
\language=\csname l@#1\endcsname
\fi
#2}}
\providecommand{\BIBdecl}{\relax}
\BIBdecl

\bibitem{odida2024evolution}
M.~Odida, ``The evolution of mobile communication: A comprehensive survey on 5g
  technology,'' \emph{J Sen Net Data Comm}, vol.~4, no.~1, pp. 01--11, 2024.

\bibitem{ashraf20245g}
S.~Ashraf, J.~A. Sheikh, A.~Ashraf\emph{,~et~al.}, ``5g millimeter wave
  technology: An overview,'' \emph{Intelligent Signal Processing and RF Energy
  Harvesting for State of art 5G and B5G Networks}, pp. 97--112, 2024.

\bibitem{xue2024survey}
Q.~Xue, C.~Ji, S.~Ma\emph{,~et~al.}, ``A survey of beam management for mmwave
  and thz communications towards 6g,'' \emph{IEEE Communications Surveys \&
  Tutorials}, pp. 1--41, 2024.

\bibitem{semiari2016downlink}
O.~Semiari, W.~Saad, and M.~Bennis, ``Downlink cell association and load
  balancing for joint millimeter wave-microwave cellular networks,'' in
  \emph{2016 IEEE global communications conference (GLOBECOM)}.\hskip 1em plus
  0.5em minus 0.4em\relax IEEE, 2016, pp. 1--6.

\bibitem{khawam2020coordinated}
K.~Khawam, S.~Lahoud, M.~El~Helou\emph{,~et~al.}, ``Coordinated framework for
  spectrum allocation and user association in 5g hetnets with mmwave,''
  \emph{IEEE Transactions on Mobile Computing}, vol.~21, no.~4, pp. 1226--1243,
  2020.

\bibitem{liu2016user}
D.~Liu, L.~Wang, Y.~Chen\emph{,~et~al.}, ``User association in 5g networks: A
  survey and an outlook,'' \emph{IEEE Communications Surveys \& Tutorials},
  vol.~18, no.~2, pp. 1018--1044, 2016.

\bibitem{alizadeh2022reinforcement}
A.~Alizadeh and M.~Vu, ``Reinforcement learning for user association and
  handover in mmwave-enabled networks,'' \emph{IEEE Transactions on Wireless
  Communications}, vol.~21, no.~11, pp. 9712--9728, 2022.

\bibitem{liu2020user}
R.~Liu, M.~Lee, G.~Yu\emph{,~et~al.}, ``User association for millimeter-wave
  networks: A machine learning approach,'' \emph{IEEE Transactions on
  Communications}, vol.~68, no.~7, pp. 4162--4174, 2020.

\bibitem{alizadeh2020multi}
A.~Alizadeh and M.~Vu, ``Multi-armed bandit load balancing user association in
  5g cellular hetnets,'' in \emph{GLOBECOM 2020-2020 IEEE Global Communications
  Conference}.\hskip 1em plus 0.5em minus 0.4em\relax IEEE, 2020, pp. 1--6.

\bibitem{alizadeh2019load}
A.~Alizadeh and M.~Vu, ``Load balancing user association in millimeter wave
  mimo networks,'' \emph{IEEE Transactions on Wireless Communications},
  vol.~18, no.~6, pp. 2932--2945, 2019.

\bibitem{luo2014downlink}
S.~Luo, R.~Zhang, and T.~J. Lim, ``Downlink and uplink energy minimization
  through user association and beamforming in c-ran,'' \emph{IEEE Transactions
  on Wireless Communications}, vol.~14, no.~1, pp. 494--508, 2014.

\bibitem{pedersen2016flexible}
K.~I. Pedersen, G.~Berardinelli, F.~Frederiksen\emph{,~et~al.}, ``A flexible 5g
  wide area solution for tdd with asymmetric link operation,'' \emph{IEEE
  Wireless Communications}, vol.~24, no.~2, pp. 122--128, 2016.

\bibitem{xu2023energy}
K.~Xu, F.-C. Zheng, H.~Xu\emph{,~et~al.}, ``Energy efficient beamforming for
  millimeter-wave massive mimo systems under user-wise asymmetric
  uplink-downlink traffic,'' \emph{IEEE Transactions on Wireless
  Communications}, 2023.

\bibitem{kim2023distributed}
Y.~Kim, J.~Jang, and H.~J. Yang, ``Distributed resource allocation and user
  association for max-min fairness in hetnets,'' \emph{IEEE Transactions on
  Vehicular Technology}, 2023.

\bibitem{ye2013user}
Q.~Ye, B.~Rong, Y.~Chen\emph{,~et~al.}, ``User association for load balancing
  in heterogeneous cellular networks,'' \emph{IEEE Transactions on Wireless
  Communications}, vol.~12, no.~6, pp. 2706--2716, 2013.

\bibitem{tam2016joint}
H.~H.~M. Tam, H.~D. Tuan, D.~T. Ngo\emph{,~et~al.}, ``Joint load balancing and
  interference management for small-cell heterogeneous networks with limited
  backhaul capacity,'' \emph{IEEE Transactions on Wireless Communications},
  vol.~16, no.~2, pp. 872--884, 2016.

\bibitem{fang2020energy}
F.~Fang, G.~Ye, H.~Zhang\emph{,~et~al.}, ``Energy-efficient joint user
  association and power allocation in a heterogeneous network,'' \emph{IEEE
  Transactions on Wireless Communications}, vol.~19, no.~11, pp. 7008--7020,
  2020.

\bibitem{li2019user}
D.~Li, H.~Zhang, K.~Long\emph{,~et~al.}, ``User association and power
  allocation based on q-learning in ultra dense heterogeneous networks,'' in
  \emph{2019 IEEE Global Communications Conference (GLOBECOM)}.\hskip 1em plus
  0.5em minus 0.4em\relax IEEE, 2019, pp. 1--5.

\bibitem{zhang2020optimal}
Y.~Zhang, L.~Dai, and E.~W. Wong, ``Optimal bs deployment and user association
  for 5g millimeter wave communication networks,'' \emph{IEEE Transactions on
  Wireless Communications}, vol.~20, no.~5, pp. 2776--2791, 2020.

\bibitem{qiao2023joint}
Y.~Qiao, Y.~Niu, Z.~Han\emph{,~et~al.}, ``Joint optimization of resource
  allocation and user association in multi-frequency cellular networks assisted
  by ris,'' \emph{IEEE Transactions on Vehicular Technology}, 2023.

\bibitem{liu2019joint}
R.~Liu, Q.~Chen, G.~Yu\emph{,~et~al.}, ``Joint user association and resource
  allocation for multi-band millimeter-wave heterogeneous networks,''
  \emph{IEEE Transactions on Communications}, vol.~67, no.~12, pp. 8502--8516,
  2019.

\bibitem{deng2023gnn}
W.~Deng, Y.~Liu, M.~Li\emph{,~et~al.}, ``Gnn-aided user association and beam
  selection for mmwave integrated heterogeneous networks,'' \emph{IEEE Wireless
  Communications Letters}, 2023.

\bibitem{sana2020multi}
M.~Sana, A.~De~Domenico, W.~Yu\emph{,~et~al.}, ``Multi-agent reinforcement
  learning for adaptive user association in dynamic mmwave networks,''
  \emph{IEEE Transactions on Wireless Communications}, vol.~19, no.~10, pp.
  6520--6534, 2020.

\bibitem{guo2020joint}
D.~Guo, L.~Tang, X.~Zhang\emph{,~et~al.}, ``Joint optimization of handover
  control and power allocation based on multi-agent deep reinforcement
  learning,'' \emph{IEEE Transactions on Vehicular Technology}, vol.~69,
  no.~11, pp. 13\,124--13\,138, 2020.

\bibitem{chaieb2022deep}
C.~Chaieb, F.~Abdelkefi, and W.~Ajib, ``Deep reinforcement learning for
  resource allocation in multi-band and hybrid oma-noma wireless networks,''
  \emph{IEEE Transactions on Communications}, vol.~71, no.~1, pp. 187--198,
  2022.

\bibitem{wang2005interference}
L.-C. Wang, S.-Y. Huang, and Y.-C. Tseng, ``Interference analysis and resource
  allocation for tdd-cdma systems to support asymmetric services by using
  directional antennas,'' \emph{IEEE Transactions on Vehicular Technology},
  vol.~54, no.~3, pp. 1056--1069, 2005.

\bibitem{rhee2011levy}
I.~Rhee, M.~Shin, S.~Hong\emph{,~et~al.}, ``On the levy-walk nature of human
  mobility,'' \emph{IEEE/ACM Transactions on Networking}, vol.~19, no.~3, pp.
  630--643, 2011.

\bibitem{goemans1994879}
M.~X. Goemans and D.~P. Williamson, ``. 879-approximation algorithms for max
  cut and max 2sat,'' in \emph{Proceedings of the twenty-sixth annual ACM
  symposium on Theory of computing}, 1994, pp. 422--431.

\bibitem{balanis2016antenna}
C.~A. Balanis, \emph{Antenna theory: analysis and design}.\hskip 1em plus 0.5em
  minus 0.4em\relax John wiley \& sons, 2016.

\bibitem{yang2010nature}
X.-S. Yang, \emph{Nature-inspired metaheuristic algorithms}.\hskip 1em plus
  0.5em minus 0.4em\relax Luniver press, 2010.

\bibitem{shokri2015beam}
H.~Shokri-Ghadikolaei, L.~Gkatzikis, and C.~Fischione, ``Beam-searching and
  transmission scheduling in millimeter wave communications,'' in \emph{2015
  IEEE international conference on communications (ICC)}.\hskip 1em plus 0.5em
  minus 0.4em\relax IEEE, 2015, pp. 1292--1297.

\bibitem{blackard1992path}
K.~L. Blackard, M.~J. Feuerstein, T.~S. Rappaport\emph{,~et~al.}, ``Path loss
  and delay spread models as functions of antenna height for microcellular
  system design,'' in \emph{[1992 Proceedings] Vehicular Technology Society
  42nd VTS Conference-Frontiers of Technology}.\hskip 1em plus 0.5em minus
  0.4em\relax IEEE, 1992, pp. 333--337.

\bibitem{nwelih2022method}
E.~Nwelih, ``A method for real-time adaptive propagation loss modeling and
  estimation over los and nlos microcellular radio communication links,''
  \emph{Journal of Science and Technology Research}, vol.~4, no.~3, pp.
  133--140, 2022.

\bibitem{rappaport2013millimeter}
T.~S. Rappaport, S.~Sun, R.~Mayzus\emph{,~et~al.}, ``Millimeter wave mobile
  communications for 5g cellular: It will work!'' \emph{IEEE Access}, vol.~1,
  pp. 335--349, 2013.

\bibitem{bai2012using}
T.~Bai, R.~Vaze, and R.~W. Heath, ``Using random shape theory to model blockage
  in random cellular networks,'' in \emph{2012 International Conference on
  Signal Processing and Communications (SPCOM)}.\hskip 1em plus 0.5em minus
  0.4em\relax IEEE, 2012, pp. 1--5.

\bibitem{andrews2014overview}
J.~G. Andrews, S.~Singh, Q.~Ye\emph{,~et~al.}, ``An overview of load balancing
  in hetnets: Old myths and open problems,'' \emph{IEEE Wireless
  Communications}, vol.~21, no.~2, pp. 18--25, 2014.

\bibitem{munkres1957algorithms}
J.~Munkres, ``Algorithms for the assignment and transportation problems,''
  \emph{Journal of the society for industrial and applied mathematics}, vol.~5,
  no.~1, pp. 32--38, 1957.

\bibitem{von2007tutorial}
U.~Von~Luxburg, ``A tutorial on spectral clustering,'' \emph{Statistics and
  Computing}, vol.~17, pp. 395--416, 2007.

\bibitem{rangan2014millimeter}
S.~Rangan, T.~S. Rappaport, and E.~Erkip, ``Millimeter-wave cellular wireless
  networks: Potentials and challenges,'' \emph{Proceedings of the IEEE}, vol.
  102, no.~3, pp. 366--385, 2014.

\bibitem{lee2022message}
H.~Lee, J.~Park, S.~H. Lee\emph{,~et~al.}, ``Message-passing based user
  association and bandwidth allocation in hetnets with wireless backhaul,''
  \emph{IEEE Transactions on Wireless Communications}, vol.~22, no.~1, pp.
  704--717, 2022.

\bibitem{dai2022joint}
C.~Dai, K.~Zhu, Z.~Li\emph{,~et~al.}, ``Joint decoupled multiple-association
  and resource allocation in full-duplex heterogeneous cellular networks: A
  four-sided matching game,'' \emph{IEEE Transactions on Wireless
  Communications}, vol.~21, no.~8, pp. 6464--6477, 2022.

\bibitem{bai2014analysis}
T.~Bai, R.~Vaze, and R.~W. Heath, ``Analysis of blockage effects on urban
  cellular networks,'' \emph{IEEE Transactions on Wireless Communications},
  vol.~13, no.~9, pp. 5070--5083, 2014.

\end{thebibliography}

\end{document}